\documentclass[10pt, journal, compsoc]{IEEEtran}
\ifCLASSOPTIONcompsoc
  \usepackage{cite}
\else
  \usepackage{cite}
\fi
\usepackage[pdftex]{graphicx}
\usepackage{amsmath}
\usepackage{algorithmic}
\usepackage{romannum}
\usepackage{array}
\ifCLASSOPTIONcompsoc
  \usepackage[caption=false,font=footnotesize,labelfont=sf,textfont=sf]{subfig}
\else
  \usepackage[caption=false,font=footnotesize]{subfig}
\fi
\usepackage{fixltx2e}
\usepackage{amsthm}
\usepackage{amssymb}
\usepackage{dblfloatfix}
\usepackage{url}
\usepackage[linesnumbered,ruled]{algorithm2e}
\usepackage{soul} 
\usepackage{xcolor}

\newcommand{\old}[1]{\st{}}

\newtheorem{lem}{Lemma}

\DeclareMathOperator*{\argmax}{argmax}

\begin{document}
\title{Small Profits and Quick Returns: An Incentive Mechanism Design for IoT-based Crowdsourcing under Continuous Platform Competition}

\author{Duin Baek,~\IEEEmembership{Student Member,~IEEE,}
        Jing Chen,
        and~Bong~Jun~Choi,~\IEEEmembership{Member,~IEEE}
\IEEEcompsocitemizethanks{\IEEEcompsocthanksitem D. Baek, J. Chen and B.J. Choi are with the Department of Computer Science, Stony Brook University, New York,
NY, 30332.\protect\\
Corresponding Author: B.J. Choi, E-mail: BongJun.Choi@stonybrook.edu}
\thanks{Manuscript received December 21}}


\IEEEtitleabstractindextext{%
\begin{abstract}
Crowdsourcing can be applied to the Internet-of-Things (IoT) systems to provide more scalable and efficient services to support various tasks. As the driving force of crowdsourcing is the interaction among participants, various incentive mechanisms have been proposed to attract and retain a sufficient number of participants to provide a sustainable crowdsourcing service. However, there exist some gaps between the modeled entities or markets in the existing works and those in reality: 1) \textit{dichotomous} task valuation and workers' punctuality, and 2) crowdsourcing service market \textit{monopolized} by a platform. To bridge those gaps of such impractical assumption, we model workers' heterogeneous punctuality behavior and task depreciation over time. Based on those models, we propose an Expected Social Welfare Maximizing (ESWM) mechanism that aims to maximize the expected social welfare by attracting and retaining more participants in the long-term, i.e., multiple rounds of crowdsourcing. In the evaluation, we modeled the continuous competition between the ESWM and one of the existing works in both short-term and long-term scenarios. Simulation results show that the ESWM mechanism achieves higher expected social welfare and platform utility than the benchmark by attracting and retaining more participants. Moreover, we prove that the ESWM mechanism achieves the desirable economic properties: individual rationality, budget-balance, computational efficiency, and truthfulness.
\end{abstract}

}

\maketitle

\IEEEdisplaynontitleabstractindextext
\IEEEpeerreviewmaketitle

\IEEEraisesectionheading{\section{Introduction}}
\label{sec:introduction}
\IEEEPARstart{T}{he} Internet-of-Thing (IoT) is one of the fastest growing research and industry areas involving vast amount and variety of computing devices, which creates the availability of massive data generation and exchange among the devices. Annually, the IoT devices grow in both quantity and quality (capability). Quantitatively, the number of IoT devices installed or deployed worldwide is expected to reach 75.44 billion in 2025 according to IHS Markit survey conducted in 2016 \cite{iot}. Qualitatively, smartphones, one of the most representative types of IoT devices, are now equipped with powerful CPU, GPU, RAM, and storage, e.g., Octa-core (4x2.8 GHz Kryo 385 Gold \& 4x1.7 GHz Kryo 385 Silver), 4 GB RAM in Samsung Galaxy S9, which are comparable to those of many laptops and desktops several years ago, as well as various sensors such as accelerometer, gyroscope, iris, and fingerprint scanner. Utilizing such powerful and numerously distributed IoT devices, we can provide a service where people can request IoT devices to collect massive amount of sensor data or to solve computationally complex problem inexpensively using their idle computing power in a manner of \textit{crowdsourcing} in which tasks that were traditionally completed by appointed agents are now outsourced to an undefined large group of crowd. Using the concept of crowdsourcing, many applications have been proposed: LiveCompare \cite{livecompare} for grocery price comparison, GreenGPS \cite{greengps} to allow drivers to find the most fuel-efficient route, LiFS \cite{lifs} for indoor localization, and so on \cite{noisetube, healthmap, bus}.   

However, such IoT-based crowdsourcing services are viable only when IoT device users actively participate in the crowdsourcing system (not only to use the services but also to provide their resources). Thus, we require an incentive mechanism that properly rewards resource providers for their resources, which will induce them to participate in the crowdsourcing system. In the literature, various incentive mechanisms for crowdsourcing have been proposed to encourage user collaboration \cite{Yang, quality, liu, Lee, Lin, budget, team1, truthful, tullock, distributed1}. 
However, despite their well-defined system models, there exist some gaps between the entities or markets modeled in the existing works and those in reality. In the crowdsourcing system, there are three kinds of entities: requesters, workers, and platforms. A \textit{requester} asks a platform to assign his/her task to a worker, while a \textit{worker} aims to complete the task to receive a reward. In the IoT-based crowdsourcing system, IoT device users can be both requesters and workers. A \textit{platform} plays as an auctioneer to mediate between requesters and workers.  

In terms of requesters, the existing works have a dichotomous task valuation model. That is, a task valuation immediately collapses to zero, like a step function, after its deadline. However, such dichotomous model is not general enough to include the task \textit{depreciation} case where the value of a task is fully preserved until its deadline and depreciates in proportion to the time past after the deadline \cite{depreciation}. Depending on the level of task depreciation over time, it may be beneficial to accept some late tasks results rather than requesting them again. Accordingly, such task depreciation can also affect the payment policy of incentive mechanisms. In the existing works that only considered fixed and binary task valuation model, workers are rewarded based on the fixed payment policy. However, to encompass general cases where tasks depreciate after their deadline and workers can still submit their task result, the reward for workers should be decided based on the task valuation achieved at the moment workers submit their assigned task result.      

Stemming from the dichotomous task valuation model, workers in the literature are modeled to show only binary behavior in terms of punctuality. In other words, workers either submit their assigned task results in time or not submit. However, considering the aforementioned general cases where tasks depreciate over time and late task results are accepted, workers' punctuality should be considered in diverse aspects. In practice, workers can have heterogeneous level of punctuality. That is, some workers may be always punctual, while others may be usually punctual but sometimes late. Even among the late workers, some may submit their task results slightly after the deadline while others submit far after the deadline. Thus, such heterogeneous punctuality should be considered to better capture the general behavior of users in the crowdsourcing market.  

Moreover, the limitation of the crowdsourcing market model in the existing works is that a platform explicitly or implicitly \textit{monopolizes} the crowdsourcing market alone. In the literature, mechanisms were separately evaluated under the same condition such as the number of participants. In addition, the crowdsourcing markets in the literature are \textit{static}, which means that they do not involve or consider movements of participants over time. However, in reality, several platforms will compete with each other in a crowdsourcing market to attract more participants. Accordingly, each participant will decide which crowdsourcing platform to join depending on the payment policy of each platform. This results in dynamic movements of participants in the market over time.

To bridge those gaps in the literature, this paper makes the following main contributions:
\begin{enumerate}
	\item We design a heterogeneous time-varying task valuation model that encompasses task depreciation.
	\item We design a behavior model of workers that captures their stochastic punctuality in completing their assigned tasks.
	\item We design an incentive mechanism that aims to maximize the expected social welfare  in the long term by attracting and retaining more participants.
	\item We model the dynamic competition over time between crowdsourcing service platforms which involves dynamic movements of participants between the platforms.
\end{enumerate}

The rest of this paper is organized as follows. In Section \ref{part:1}, we provide a literature review of the existing mechanisms for crowdsourcing. In Section \ref{part:2}, we present our system models on requesters, workers, and a platform. Based on the models, we formulate the expected social welfare maximizing problem in Section \ref{part:3}. In Section \ref{part:4}, we propose our Expected Social Welfare Maximizing (ESWM) mechanism that selects appropriate requester-worker pairs considering heterogeneity in task depreciation speed and workers' punctuality. In Section \ref{part:5}, we evaluate the performance of our ESWM mechanism. In Section \ref{part:6}, we conclude our paper.

\section{Related Work}
\label{part:1}
In this section, we review the state-of-the-art research works on incentive mechanisms for crowdsourcing. Yang et al. \cite{Yang} presented two general models of incentive mechanisms to motivate mobile users' participation: platform-centric model and user-centric model. D. Peng et al. \cite{quality} proposed a quality-based incentive mechanism for crowdsensing by rewarding participants proportionally to their contribution. Y. Wen et al. \cite{quality1} presented a quality-driven auction-based incentive mechanism for a Wi-Fi fingerprint-based indoor localization system. In this direction, C. Liu et al. \cite{liu} also proposed a Quality of Information (QoI)-aware incentive mechanism to maximize the quality of information. 

In the long-term view, Lee and Hoh \cite{Lee} proposed a mechanism, called RADP-VPC, that provides long-term incentives to participants to maintain participants and promote dropped ones to participate again. Similarly, L. Gao et al. \cite{Lin} proposed a mechanism to provide long-term incentives to participants to achieve the maximum total sensing value and the minimum total sensing cost. 

Focusing on dynamic crowdsensing where participants arrive in an online manner, Zhao et al. \cite{budget} presented two online incentive mechanisms using a multiple-stage sampling-accepting process. Similarly, Y. Wei et al. \cite{truthful} proposed an online incentive mechanism to maximize the number of matched pairs of participants and crowdsourcing service users when participants and service users are dynamically changing.

J. Sun et al. \cite{behavior} proposed a behavior-based incentive mechanism for crowdsensing with budget constraint. This work aims to achieve both the extensive user participation and high quality sensing data submission, based on users’ behavior abilities. S. Ji et al. \cite{tremble} presented an incentive mechanism for mobile phones with uncertain sensing time. In order to address the sensing time uncertainty problem, this work modeled the problem as a perturbed Stackelberg game where mobile phone users sensing task times may be different from their original plans. T. Luo et al. \cite{allpay} proposed an incentive mechanism for heterogeneous crowdsourcing using all-pay contests where workers have not only different type information (abilities and costs), but also the different beliefs (probabilistic knowledge) about their respective type information. In this work, the belief is modeled as a probability distribution. 

X. Zhang et al. \cite{team1} proposed truthful incentive mechanisms for the case where participants' cooperation may be required to finish a job. Using Tullock contest, T. Luo et al. \cite{tullock} designed an incentive mechanism to maximize the crowdsourcing service user’s profit. L. Duan et al. \cite{distributed1} proposed an incentive mechanism for smartphone collaboration in distributed computing using contract theory under two different scenarios having complete and incomplete information of participants.

However, the heterogeneity of task depreciation over time has not been addressed in the literature. Task depreciation models over time in the existing works are \textit{dichotomous} as the value of tasks immediately drop to zero after the given deadline. In other words, there can be only two kinds of task valuation, either full or null. Such dichotomous model overlooks the cases where task valuation remains valid even after the deadline, though depreciating in proportion to the amount of time past the deadline, which can be observed in practice and the literature \cite{depreciation}.

Moreover, the heterogeneity of workers' \textit{punctuality} levels has not been addressed. In real crowdsourcing systems, each worker's priority on its assigned task can vary, e.g., some prioritize on the assigned tasks while others prioritize on their own tasks. In the existing works, workers are assumed to have dichotomous behaviors. Moreover, many works implicitly or explicitly assumed that all selected workers meet their deadlines. However, coupled with the depreciation of task value after the deadline, workers' heterogeneous punctuality levels can make difference in the realized task valuation at the task result submission time. For instance, a provider with a relatively higher punctuality level is expected to achieve a higher task valuation than the one with a lower punctuality level. Besides, the partial task valuation achieved after the deadline can vary depending on the punctuality level.


\section{System Model}
\label{part:2}
\subsection{System Overview}
In this section, we propose a system model considering task depreciation over time and workers' stochastic punctuality. A crowdsourcing system consists of a platform (an auctioneer), a set of $M$ requesters (buyers) $R=\{r_1,r_2,r_3,...,r_M\}$ identified by $j$, and a set of $N$ workers (sellers) $W=\{w_1,w_2,w_3,...,w_N\}$ identified by $i$. Then, the crowdsourcing system is extended to two platforms that compete with each other. 
As mentioned earlier, IoT devices or users can be both requesters and workers in the IoT-based crowdsourcing system. For simplicity, we assume that a requester $r_j$ can submit only one task $\Gamma_j$ to the platform and a worker $w_i$ can take only one task. As in the real world systems, a platform is assumed to have a limited capacity to handle $K$ number of task requests. The process of crowdsourcing in our system is in a form of double auction where requesters (workers) compete with each other to be selected as the winning requesters (workers) as detailed below:
\begin{enumerate}
	\item A set of requesters $R$ submit to the platform their own type information $\theta_j^r$ which contains the maximum task valuation $v^{max}_j$, deadline $t^d_j$, and etc. Similarly, a set of workers $W$ submit to the platform their own type information $\theta_i^p$ which contains their cost value $c_i$. The platform analyzes and adds the level of punctuality $\lambda_i$ to the corresponding $\theta_i^p$.  
	\item The platform selects a subset of $R$ as the winning requesters $R_s$ and a subset of $W$ as the winning workers $W_s$, and calculates temporary fees $q_j \leq v^{max}_j$ for $r_j \in R_s$ and temporary payment $p_i \geq c_i$ for $w_i \in W_s$ assuming that all $w_i \in W_s$ will be perfectly punctual. Note that the temporary $q_j$ and the payment $p_i$ will be updated based on the task valuation at the time of actual task result submission.
	\item For $R_s$ and $W_s$, the platform matches each $r_j \in R_s$ to each $w_i\in W_s$. The matching algorithm is detailed in Section \ref{part:4-matching}.
	\item When each $w_i \in R_s$ completes and submits the requested task $\Gamma_j$ at $t_i^{sub}$, the platform decides definitive fee $q'_j$ and payment $p'_i$, based on the achieved task valuation at $t_i^{sub}$. Note that $q'_j$ and payment $p'_i$ are the fee that $r_j \in R_s$ will be charged and the payment $w_i \in W_s$ will receive, respectively.        
\end{enumerate}

\subsection{Requesters}
\begin{figure}
	\centering
	\includegraphics[width=0.3\textwidth]{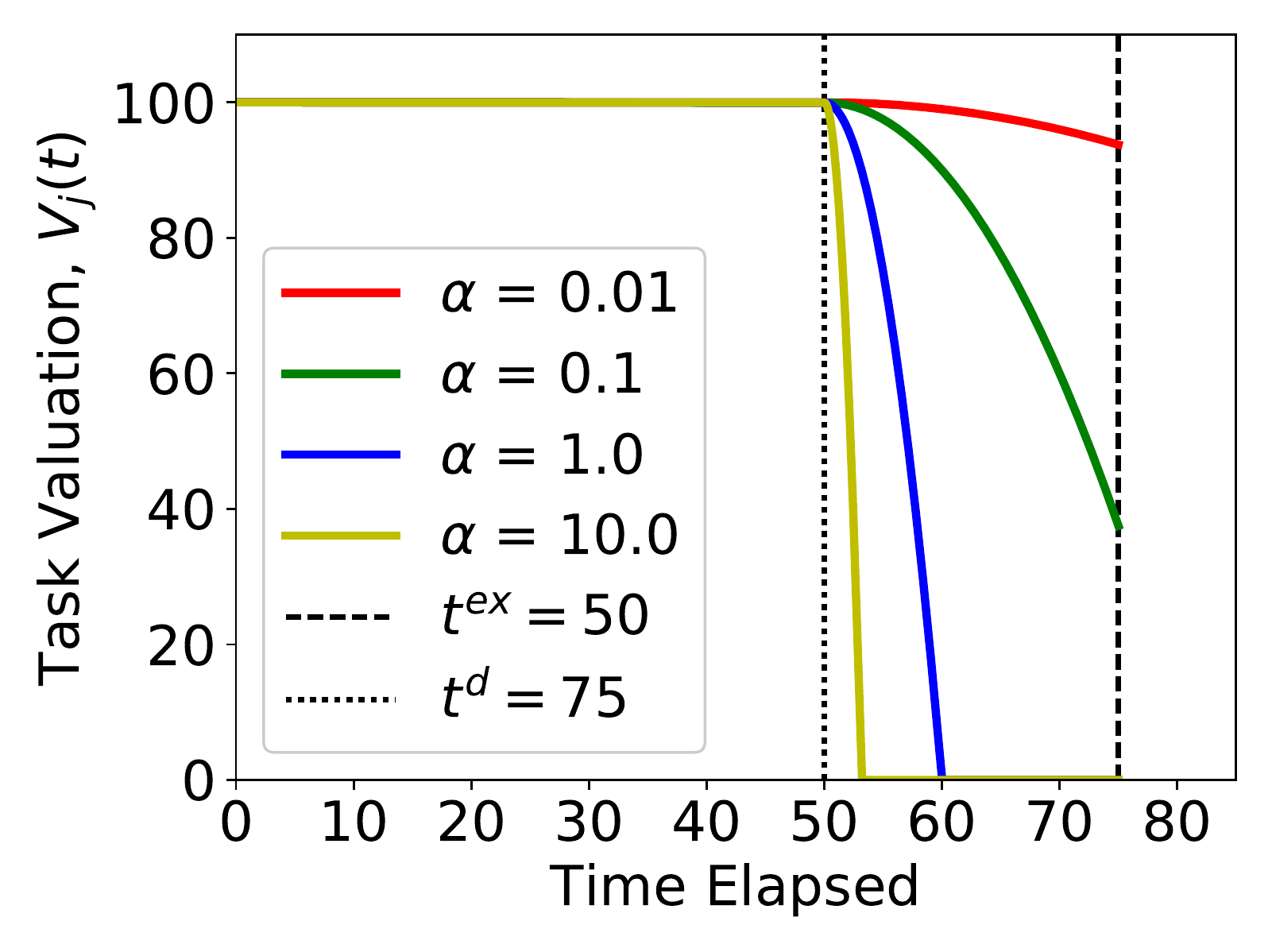}
	\caption{$v_j(t)$ with different $\alpha$ values}
	\label{fig:1}
\end{figure}

In crowdsourcing, requesters are buyers in a double auction who compete to be selected as the winners so that they can outsource their tasks to workers. For those selected $R_s$, their outsourced tasks will be completed by the winning workers. In the initial stage of the auction, each requester $r_j$ submits to the platform its own type information $\theta_j^r=(\,\Gamma_j, t_j^d, t_j^{ex}, v^{max}_j, \alpha_j )\,$. Here, $\Gamma_j$, $t_j^d$, and $t_j^{ex}$, respectively, denote a task outsourced by requester $r_j$, the deadline when the valuation of the outsourced task starts to depreciate, and the expiry time when the valuation of the task becomes null. The maximum valuation of $\Gamma_j$ is denoted as $v_j^{max}$. As each requester $r_j$ is rationally selfish, he/she will use crowdsourcing service only if the value of $\Gamma_j$ is higher than the corresponding fee $q_j$. To this extent, requesters are modeled in the same way as the existing works where tasks are heterogeneous in various aspects such as task size, bid, and etc. Such heterogeneity has been well formulated into mathematical problems and various solutions for them have been proposed. 

In addition, we consider the heterogeneity of task depreciation over time which has not been addressed in the literature as mentioned in Section \ref{part:1}. In this work, to reflect the diversity of task depreciation in crowdsourcing, we introduce $\alpha_j$ to denote the speed of task depreciation after the deadline. In practice, the speed or level of task depreciation can vary, i.e., some tasks depreciate to null right after their deadline, while others do gradually. Therefore, to reflect such heterogeneity of task depreciation, we introduce a new function in terms of elapsed time $t$ which reflects the speed of task depreciation defined as  
\begin{equation}
\label{eq:5}
v_j(t) = 
\begin{cases}
v_j^{max}, & \quad \text{if } 0 \leq t \leq t_j^d \\
\max\{ 0 , v_j^{max}-\alpha_j(t-t_j^d)^2\},  & \quad \text{otherwise}.\\
\end{cases}
\end{equation}

We use a quadratic function of elapsed time $t$ as a form of \textit{accelerated depreciation} methods where tasks are more profitable or have a higher utility during their early time period. Until the deadline $t_j^d$ of $r_j$, the task valuation $v_j(t)$ remains constant at maximum $v_j^{max}$. However, $v_j(t)$ starts to depreciate after $t_j^d$, depending on its $\alpha_j$. Tasks with small $\alpha_j$ will depreciate slowly and we may achieve the partial task valuation when completed even after $t_j^d$. In contrast, tasks with large $\alpha_j$ will depreciate too fast to achieve any valuation even though it is completed shortly after $t_j^d$. In Fig. \ref{fig:1}, we show how $v_j(t)$ changes with different $\alpha_j$ values.
Based on $v_j(t)$ and $q'_j$, we define the utility of requester $r_j$ as follows
\begin{equation}
\label{eq:1}
u_{j}^{r}= 
\begin{cases}
v_j(t_*^{sub}) - q'_j, & \quad \text{if } r_j \in R_s  \\
0 & \quad \text{otherwise},\\
\end{cases}
\end{equation}
where $t_*^{sub}$ denotes the time $\Gamma_j$ is completed. 

\subsection{Workers}
\begin{figure}
	\centering
	\includegraphics[width=0.3\textwidth]{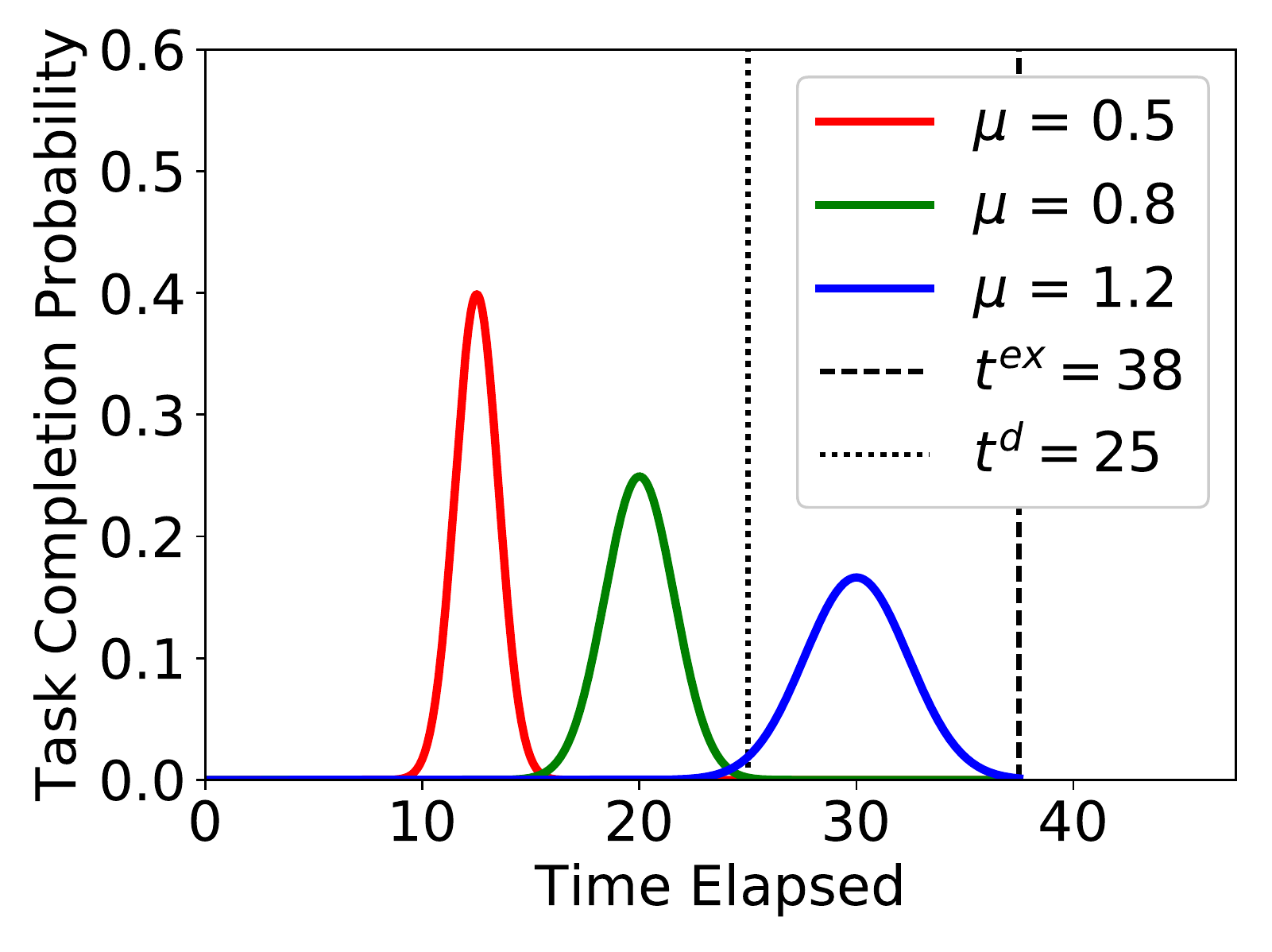}
	\caption{$f_i(t;\mu_i,\sigma_i, t_*^d, t_*^{ex})$ with different $\mu_i$ values}
	\label{fig:2}
\end{figure}
In crowdsourcing, workers are sellers in a double auction who compete to be selected as the winners. Each winning worker $w_i \in W_s$ is required to complete their requested tasks. In the process of executing assigned tasks, each worker $w_i$ will have task execution cost, $c_{i}\in \mathbb{R}_{+}$. To compensate such incurred cost, $w_i$ submits to the platform its own type information $\theta_i^p=(\,c_i)\,$ that represents the minimum ask value that $w_i$ wants to receive as the reward for completing a requested task. Because $w_i$ is rationally selfish, $w_i$ will decide to work on the requested task only if it is paid with $p'_i \geq c_i$. To this extent, workers are modeled in the same way as the existing works where the costs of workers are heterogeneous. 

On top of the heterogeneity in costs, we consider workers' heterogeneous behavior in meeting the deadline which has not been addressed in the existing works. To reflect such heterogeneous punctuality levels, we model $w_i$'s stochastic behavior for its assigned task $\Gamma_*$ as a conditional probability function parameterized by $t_*^d$ and $t_*^{ex}$, using the following truncated normal distribution 
\begin{equation}
\label{eq:7}
f_i(t;\mu_i,\sigma_i, t_*^d, t_*^{ex})= 
\begin{cases}
\frac{\frac{1}{\sigma_i}\phi(\frac{t-\mu_i t_*^d}{\sigma_i}) }{\Phi(\frac{t_*^{ex}-\mu_i t_*^d}{\sigma_i}) - \Phi(\frac{-\mu_i}{\sigma_i})}, & 0 \leq t \leq t_*^{ex}  \\
0, & \text{otherwise},\\
\end{cases}
\end{equation}
\begin{equation}
\label{eq:8}
\text{where } \phi(\zeta)=\frac{1}{\sqrt{2\pi}}\exp(-\frac{1}{2}\zeta^2),
\end{equation}
where (\ref{eq:8}) and $\Phi(\cdot)$ denote the probability density function of the standard normal distribution and the cumulative density function, respectively. Note that $\mu_i t^d_*$ is the mean of the truncated normal distribution, which can be derived from users' log data. Note that the existing works \cite{behavior1,behavior2,behavior} demonstrate that we can analyze user behaviors, using log analysis. In this work, we assume $\sigma=2\mu$ for simplicity. Such linear relationship between $\mu$ and $\sigma$ is assumed to more clearly differentiate workers' behaviors. For a newly participating $w_i$, the platform sets its $\mu_i$ to 1.  

Based on the behavior model, workers stochastically submit their requested task results to the platform. In other words, there is a high probability of $w_i$ submitting the task results around $\mu_it_*^d$. To numerically represent such workers' stochastic behavior, we use a punctuality coefficient $\lambda_i$, which is equal to $1/\mu_i$. The platform appends $\lambda_i$ to each $\theta_i^p$ resulting in $\theta_i^p = (a_i, \lambda_i)$. A user's behaviors with different $\mu$ values for the same task are shown in Fig. \ref{fig:2}. Based on the payment $p'_i$ and the incurred cost $c_i$, we define the utility of worker $w_i$ as follows
\begin{equation}
\label{eq:2}
u_{i}^{p} =
\begin{cases}
p'_{i}-c_{i}, & \quad \text{if } w_i \in W_s\\
0  & \quad \text{otherwise}.\\
\end{cases}
\end{equation}

\subsection{Platform}
In crowdsourcing systems, a platform acts as an auctioneer to select the winners in both requesters (buyers) and providers (sellers), and match each winning requester to a winning worker. Reflecting servers have finite capacity, we assume the platform has a limited capacity $K$ to handle outsourcing requests.

For a platform, the sum of fees from $R_s$ and the sum of payments to $W_s$ are its revenue and expenditure, respectively. Given $R_s$ and $W_s$, we define the utility of a platform as follows
\begin{equation}
\label{eq:3}
u_{0}= \sum_{r_j \in R_s} q'_j - \sum_{w_i\in W_s} p'_i.
\end{equation}

\subsection{Expected Social Welfare}
In the existing works, to evaluate the performance of crowdsourcing services, the system-wise social welfare is calculated as follows
\begin{equation}
\label{eq:20}
\sum_{r_j \in R_s}v_j^{max} - \sum_{w_i \in W_s} c_i.
\end{equation} 

However, when tasks depreciate after their deadline at various speed and workers' behaviors are stochastic, the expected social welfare can better evaluate the performance of crowdsourcing services than the simple calculation in (\ref{eq:20}). Thus, to reflect workers' stochastic behaviors and potential task depreciation in the performance evaluation, we define the system-wise expected social welfare (ESW) as 
\begin{equation}
ESW = \sum_{r_j \in R_s}\mathbf{E}_i(v_j(t)) - \sum_{w_i \in W_s}c_i,
\label{eq:4}
\end{equation}
where $\mathbf{E}_i(v_j(t))$ is the expected valuation of $\Gamma_j$ when assigned to $w_i$. Note that in this work, we assume that the cost $c_i$ is constant whenever $w_i$ completes its assigned task. We define $\mathbf{E}_i(v_j(t))$ as   
\begin{multline}
\label{eq:10}
\mathbf{E}_i(v_j(t)) = \int_{0}^{t^{ex}_j}v_j(t) f(t;\mu_i,\sigma_i, t_j^d, t_j^{ex})dt = \int_{0}^{t^d_j}v_j^{max}  \\ 
f(t;\mu_i,\sigma_i, t_j^d, t_j^{ex}) dt + 
\int_{t^d_j}^{t^{ex}_j}v_j(t) f(t;\mu_i,\sigma_i, t_j^d, t_j^{ex}) dt,
\end{multline}
where $v_j(t)$ remains constant until $t_j^d$. Here, $\mathbf{E}_i(v_j(t))$ consists of two parts: 1) pre-deadline, and 2) post-deadline. The pre-deadline part and the post-deadline part represent the expected task valuation until the deadline and that after the deadline, respectively. Note that (\ref{eq:20}) is a special case of (\ref{eq:4}) where the integral of $f_i(t;\mu_i,\sigma_i, t_j^d, t_j^{ex})$ from 0 until $t_j^d$ is 1 because $\mathbf{E}_i(v_j(t))$ in this case is calculated as  
\begin{multline}
\label{eq:9}
\mathbf{E}_i(v_j(t))=\int_{0}^{t^{ex}_j}v_j(t)f_i(t;\mu_i,\sigma_i, t_j^d, t_j^{ex})dt =\\\int_{0}^{t^d_j} v_j^{max} f_i(t;\mu_i,\sigma_i, t_j^d, t_j^{ex}) dt= v_j^{max}.
\end{multline}
Note that when the assumption of perfect punctuality does not hold, the expected valuation of (\ref{eq:10}) can be less than $v_j^{max}$, i.e., $\mathbf{E}_i(v_j(t)) \leq v_j^{max}$, due to the task results submitted after the deadline. 

\subsection{Desirable Economic Properties}
In this work, we aim to design an incentive mechanisms for crowdsourcing that satisfies the following four desirable economic properties: 1) individual rationality, 2) budget-balance, 3) computational efficiency, and 4) truthfulness. Each is described as below.
\subsubsection{Individual Rationality}
An incentive mechanism is individually rational if both requesters (buyers) and workers (sellers) have non-negative utility when their true valuation and cost are reported. 
\subsubsection{Budget-balance}
An incentive mechanism is budget-balanced if the platform has non-negative utility at the end of the auction. That is $\sum_{r_j \in R_s}q_j - \sum_{w_i \in W_s}p_i \geq 0$.  
\subsubsection{Computational Efficiency}
An incentive mechanism is computationally efficient if it runs in polynomial time. 
\subsubsection{Truthfulness}
An incentive mechanism is truthful if neither requester nor worker can increase its utility by submitting the false task valuation or cost information. In other words, submitting the true valuation or cost information is a dominant strategy for all participants.

\section{Expected Social Welfare Maximizing Problem}
\begin{figure}
	\centering
	\subfloat[Running Time \label{fig:running_time}]{%
		\resizebox{0.24\textwidth}{!}{\includegraphics[width=\linewidth]{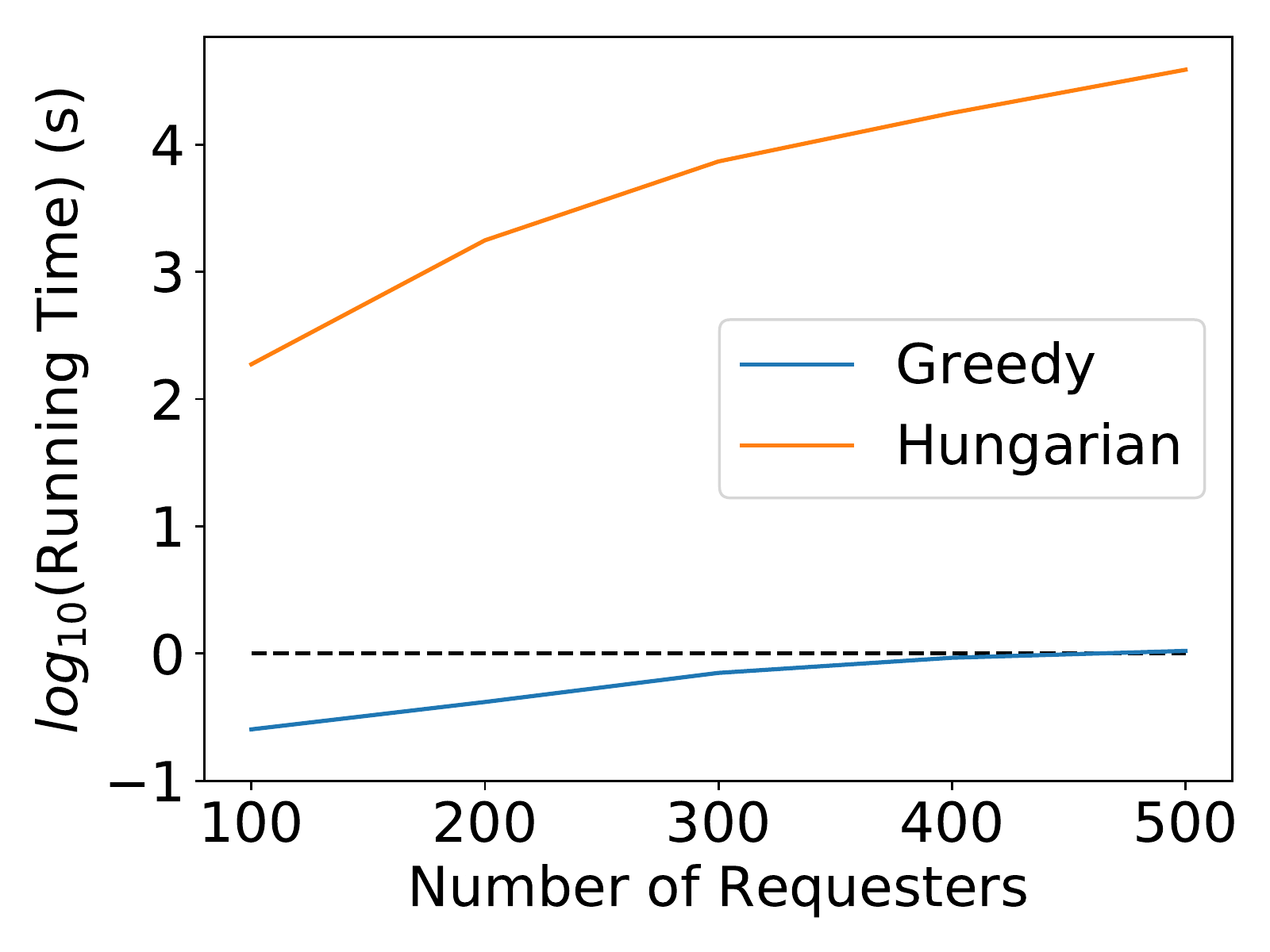}}}
	\subfloat[Expected Social Welfare \label{fig:expected}]{%
		\resizebox{0.24\textwidth}{!}{\includegraphics[width=\linewidth]{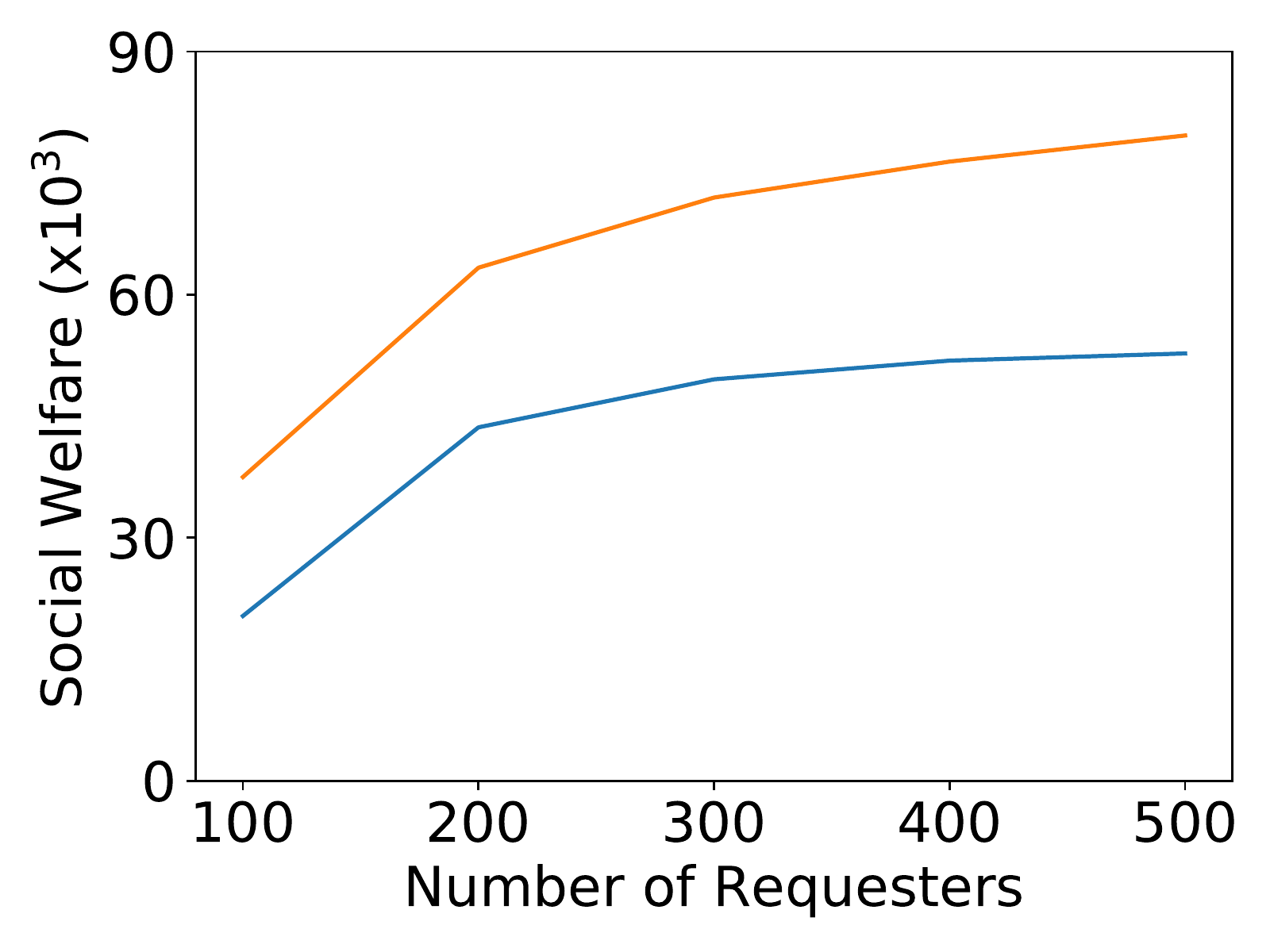}}}
	\caption{Hungarian Algorithm VS Greedy Algorithm} 
	\label{fig:optimal}
\end{figure}
\label{part:3}
In our system, the objective of the platform is to find the optimal requester-worker matches that maximize the expected social welfare. As one requester is matched with only one worker, all possible combinations of $ (r_j, w_i)$ pairs can be defined as a matrix $L$ below
\begin{equation}
L = 
\begin{bmatrix}
l_{11}  &   l_{12}  & \dots   &   l_{1\left|W\right|}       \\
l_{21}  &   l_{22}  & \dots   &   l_{2\left|W\right|}       \\
\vdots  &  \vdots   & \vdots  &   \vdots  \\   
l_{\left|R\right|1} & l_{\left|R\right|2} & \dots & l_{\left|R\right|\left|W\right|} \\
\end{bmatrix}
,
\end{equation}
where each element $l_{ji} = 1$ only if $r_j \in R$ and $w_i \in W$ are matched, otherwise 0. Thus, based on the $L$ matrix and (\ref{eq:4}), we can formulate the expected social welfare maximizing problem to find the optimal requester-provider matches $L^*$ as follows       
\begin{equation}
\label{eq:15}
L^*=\argmax_{L} \sum_{r_j \in R}\sum_{w_i \in W} (\mathbf{E}_i(v_j(t)) - c_i) l_{ji},  
\end{equation}
subject to 
\begin{equation}
\tag{12.a}
\label{eq:16}
l_{ji} \in \{0, 1\},\quad \forall l_{ji} \in L,
\end{equation}
\begin{equation}
\tag{12.b}
\label{eq:13}
\sum_{r_j \in R}\sum_{w_i \in W} l_{ji} \leq K,
\end{equation}
\begin{equation}
\tag{12.c}
\label{eq:14}
\sum_{w_i \in W}l_{ji} \leq 1,\quad \forall r_j \in R,
\end{equation}
\begin{equation}
\tag{12.d}
\label{eq:12}
\sum_{r_j \in R}l_{ji} \leq 1,\quad \forall w_i \in W.
\end{equation}

The objective function (\ref{eq:15}) is a combinatorial optimization problem to select the optimal requester-provider pairs that maximize the expected social welfare defined as (\ref{eq:4}) subject to constraints (\ref{eq:13}), (\ref{eq:14}), and (\ref{eq:12}). As mentioned in the definition of a $L$ matrix, $l_{ji}$ in (\ref{eq:16}) is a binary variable to indicate whether $r_j \in R$ and $w_i \in W$ are paired or not. Constraint (\ref{eq:13}) states that the platform has a limited capacity to handle $K$ task requests. In an ideal case where the platform can manage all the incoming task requests, $K$ in (\ref{eq:13}) can be set to $\left|R\right|$. Constraint (\ref{eq:14}) and (\ref{eq:12}) state that each requester can be matched to only one worker and vice versa.

In the ideal system where all the task requests can be handled, the optimal requester-provider matches $L^*$ can be obtained via the Hungarian algorithm \cite{hungarian}. Guaranteeing to solve the assignment problem in polynomial time, the Hungarian algorithm, also known as Munkres assignment algorithm, can find the optimal requester-worker pairs in $\mathcal{O}(n^3)$ when $n$ is $\max \{\left|R\right|, \left|W\right|\}$. However, despite its guaranteed polynomial running time, the running time of the Hungarian algorithm significantly increases to solve (\ref{eq:12}). To test the feasibility of the Hungarian algorithm in practice, we implemented it in Ubuntu 18.04.1 LTS equipped with Intel® Xeon(R) CPU E5-2630 @ 2.30GHz (24 cores) and 36GB RAM. In the test, we varied $\left|R\right|$ from 100 to 500 in an increment of 100, setting $K$ and $\left|W\right|$ to 100 and twice of $\left|R\right|$, respectively. Note that we first find the optimal pairs including all the requesters, and then select the top $K$ pairs. As shown in Fig \ref{fig:running_time}, there exists a significant difference in the time required to find the K optimal pairs. For instance, the Hungarian algorithm completes in approximately 190 seconds, while a greedy algorithm ends in 0.25 second. Moreover, such time difference widens as $\left|R\right|$ increase: the Hungarian algorithm requires almost half a day (38711 seconds $\approx$ 11 hours) to find the optimal pairs for 500 requesters, while the greedy approach-based algorithm requires approximately a second to complete (1.05 second). 
Consequently, it is not feasible to deploy the Hungarian algorithm to find the optimal requester-worker pairs in the growing IoT networks where instantaneous interactions between numerous IoT devices are demanded, despite its optimal return in social welfare as shown in Fig. \ref{fig:expected}.  

In addition, the native Hungarian algorithm does not guarantee the aforementioned desirable economic properties (individual rationality, budget-balance, and truthfulness) which are essential to sustain the crowdsourcing service despite achieving higher expected social welfare. Thus, to address this limitation, we propose an expected social welfare maximizing mechanism (ESWM) that is based on a greedy algorithm to heuristically obtain the locally optimal solution. Considering heterogeneity in task depreciation speed and workers' punctuality, the ESWM selects appropriate requester-worker pairs in a polynomial time (within a couple of seconds). Note that when the platform handles a large number of tasks in real-world IoT systems, the real-time response is essential for the practical deployment of the algorithm. Unlike the existing works, our ESWM aims to achieve a higher social welfare or platform utility in long-term view by attracting and retaining more participants, rather than attempting to simply maximize the platform utility in a given round of auction. In addition, the ESWM achieves individual rationality, budget-balance, computational efficiency, and truthfulness. 


\section{ESWM Mechanism}
\label{part:4}
\begin{figure}
	\centering
	\includegraphics[width=0.5\textwidth]{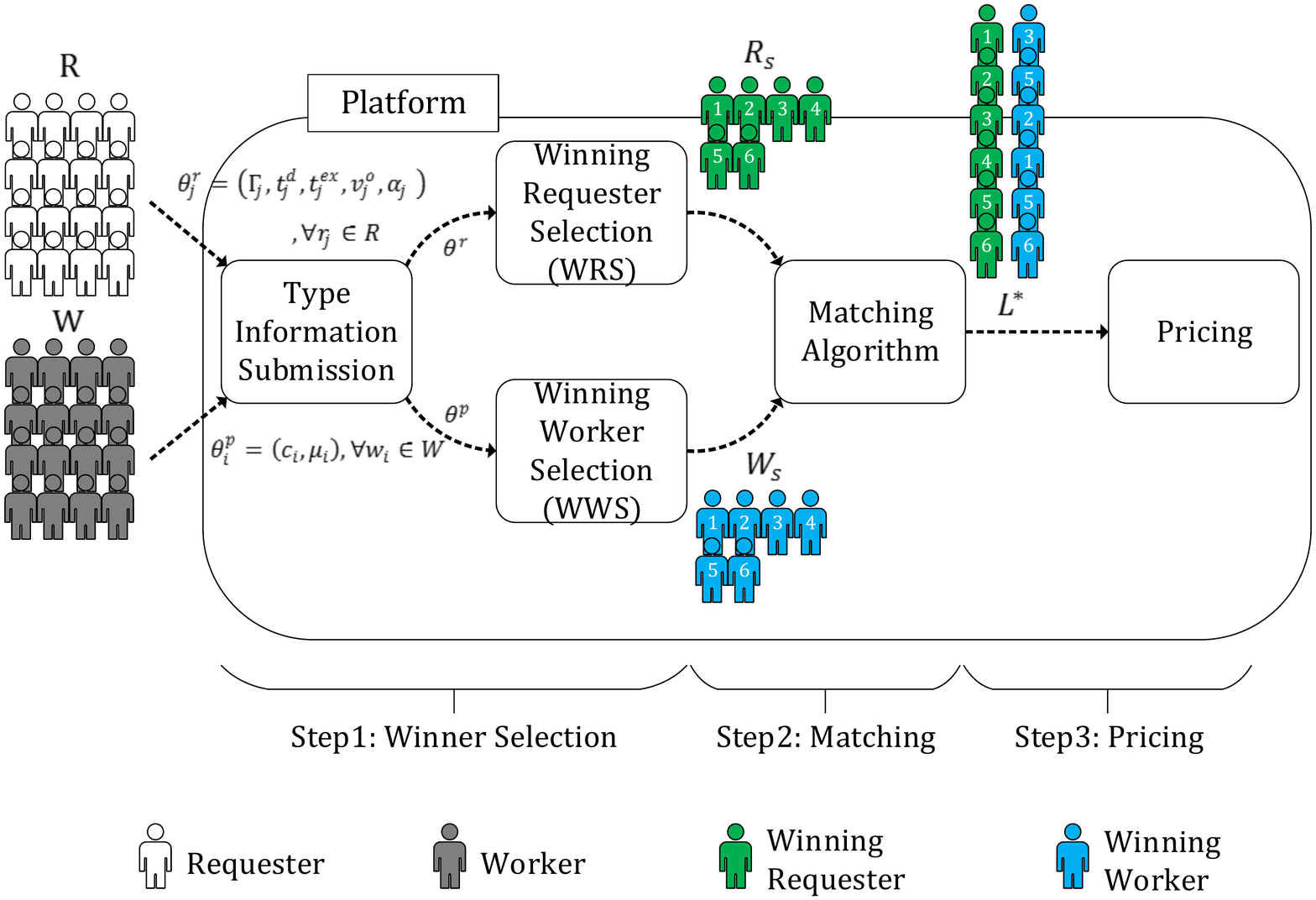}
	\caption{Workflow of H-ESWM Mechanism}
	\label{fig:7}
\end{figure}
In this section, we propose an expected social welfare maximizing (ESWM) mechanism, which largely consists of 3 main steps: 1) winner selection step, 2) matching step, and 3) pricing step. In Fig. \ref{fig:7}, the overall workflow of H-ESWM mechanism is shown.

\subsection{Winner Selection Step}        
In the initial stage of each double auction, a set of requester $R$ and a set of provider $W$ submit to the platform their type information $\theta^r$ and $\theta^p$, respectively. Ultimately, the platform aims to select the same number of winners from both $R$ and $W$. As the platform is assumed to have a limited capacity $K$ to handle task requests, the maximum numbers of $R_s$ and $W_s$ are set to $K$.  
\subsubsection{Winning Requester Selection Algorithm (WRSA)}
The winner selection criterion in the WRSA is straightforward; select requesters whose task valuation is high and slowly depreciates the task valuation after the deadline. The platform iteratively selects a requester with the maximum ratio of $v_j^o/(\alpha_j^{\beta}\left|\Gamma_j\right|)$ among $R$ as a winning requester (line 3-4) as detailed in Algorithm \ref{WRA}. Note that $\beta$ is a tunable parameter to adjust weight between $v_j^o/\left|\Gamma_j\right|$ and $\alpha_j$. 
By controlling $\beta$, we can decide the priority order between unit task valuation and depreciation speed in the winning requester selection. The effect of $\beta$ on the platform performance will be discussed in the later section.
The selection process repeats until the number of $R_s$ reaches $K+1$ or every $r_j\in R$ is selected. Among $R_s$, the platform chooses a requester with the minimum ratio of $v_j^o/(\alpha_j^{\beta}\left|\Gamma_j\right|)$ as a threshold requester, and excludes it from $R_s$ (line 9-10).

\begin{algorithm}
	\caption{Winning Requester Selection Algorithm (WRSA)}
	\label{WRA}
	
	\SetKwInOut{Input}{Input}
	\SetKwInOut{Output}{Output}
	
	\Input{$R, K$, $\beta$}
	\Output{$R_s, r_{th}$}
	
	$R_s \leftarrow \emptyset$\; 
	\While{\texttt{$\left| R_s \right| \neq K+1$}}
	{   
		$r^* \leftarrow  \arg\underset{r_j \in R}\max \frac{1}{\alpha_j^{\beta}}\frac{v_j^o}{\left|\Gamma_j \right|}$\;
		$R_s \leftarrow R_s \cup \{r^*\}$,
		$R \leftarrow R \setminus \{r^*\}$\;
		
		\If{$\left|R\right|=0$}
		{
			\textbf{break}\;
		}
	}
	$r_{th} \leftarrow  \arg\underset{r_j \in R_s}\min  \frac{1}{\alpha_j^{\beta}}\frac{v_j^o}{\left|\Gamma_j \right|}$\;  
	$R_s \leftarrow R_s\setminus\{r_{th}\}$\; 
	
	\Return{$R_s, r_{th}$}
\end{algorithm}

\subsubsection{Winning Worker Selection Algorithm (WWSA)}
The criterion and selection process of the WWSA are similar to those of the WRSA. The WRSA selects providers who have low ask values and high probability to meet the deadline. In each iteration, the platform selects a worker with a minimum ratio of $c_i/\lambda_i^{\beta}$ among $W$ as a winning worker (line 3-4). As in the WRSA, $\beta$ is a tunable parameter value to adjust where the platform puts more weight on between $c_i$ and $\lambda_i$ for the winner selection. By using the ratio, it tries to select a worker with high punctuality as well as a low-cost value. The platform repeats such selection process until the size of $W_s$ reaches $K+1$ or every $w_i\in W$ is selected. Then, the platform chooses a worker with a maximum ratio of $c_i/\lambda_i^{\beta}$ among $W_s$ as a threshold worker, who is ruled out from $W_s$ (line 9-10). 
\begin{algorithm}
	\caption{Winning Worker Selection Algorithm (WWSA)}
	\label{WPA}
	
	\SetKwInOut{Input}{Input}
	\SetKwInOut{Output}{Output}
	
	\Input{$W, K$, $\beta$}
	\Output{$W_s, w_{th}$}
	
	$W_s \leftarrow \emptyset$\;	
	\While{$\left|W_s\right|\neq K+1$}
	{
		$w^* \leftarrow \arg \underset{w_i \in W}{\min} \frac{c_i}{\lambda_i^{\beta}}$\;
		$W_s \leftarrow W_s \cup \{w^*\}$, 
		$W \leftarrow W \setminus \{w^*\}$\;		
		
		\If{$\left|W\right|=0$}
		{
			\textbf{break}\;
		}
	}
	$w_{th} \leftarrow  \arg\underset{w_i \in W_s}\max \frac{c_i}{\lambda_i^{\beta}}$\;  
	$W_s \leftarrow W_s \setminus \{w_{th}\}$\;	
	\Return{$W_s, w_{th}$}
\end{algorithm}

\subsection{Matching Step}
\label{part:4-matching}
In the matching step, the platform iteratively matches $R_s$ to $W_s$ in such a way that unmatched $r_j \in R_s$ with the maximum $v^o_j /({\alpha_j^{\beta}}\left|\Gamma_j\right|)$ is assigned to unmatched $w_i \in W_s$ with the minimum $c_i / {\lambda_i^{\beta}}$ as detailed in Algorithm \ref{Matching1}. By such matching criterion, the platform pairs a $r_j \in R_s$ whose task remains in high valuation even after its deadline with a $w_i$ who is cost-effective and punctual.    

Before the matching of $R_s$ and $W_s$, the platform first inspects whether both $R_s$ and $W_s$ can be one-to-one matched and replaces, if necessary, either $r_{th}$ or $w_{th}$, by running the trimming algorithm (TA) (line 4) as detailed in Algorithm \ref{Trimming}.
\begin{algorithm}
	\caption{Trimming Algorithm (TA)}
	\label{Trimming}
	
	\SetKwInOut{Input}{Input}
	\SetKwInOut{Output}{Output}

	\Input{$R_s, W_s, r_{th}, w_{th}$, $\beta$}
	\Output{$R_s, W_s, r_{th}, w_{th}, Q, P$}
	
	\text{Assume $\frac{v^o_1}{\alpha_1^{\beta}\left|\Gamma_1\right|}\geq \dots \geq \frac{v^o_{\left|R_s\right|}}{\alpha_{\left|R_s\right|}^{\beta} \left|\Gamma_{\left|R_s\right|}\right|}\geq \frac{v^o_{th}}{\alpha_{th}^{\beta}\left|\Gamma_{th}\right|}$}\;
	\text{Assume $\frac{c_1}{{\lambda_1}^{\beta}}\leq \dots \leq \frac{c_{\left|W_s\right|}}{{\lambda_{\left|W_s\right|}^{\beta}}} \leq \frac{c_{th}}{{\lambda_{th}^{\beta}}}$}\;
	$Q \leftarrow \emptyset$, $P \leftarrow \emptyset$\;
	\If{$\left|R_s\right| < \left|W_s\right|$}
	{
		$w_{th} \leftarrow W_s[\left|R_s\right|+1]$\;
		$W_s \leftarrow W_s[1 \sim \left|R_s\right|]$\;
	}
	\ElseIf{$\left|R_s\right| > \left|W_s\right|$}
	{
		$r_{th} \leftarrow R_s[\left|W_s\right|+1]$\;
		$R_s \leftarrow R_s[1 \sim \left|W_s\right|]$\;
	}
	$q_j \leftarrow 0, \forall r_j \in R_s$\; 
	$p_i \leftarrow 0, \forall w_i \in W_s$\; 
	
	\For{$(r_j, w_i) \in (R_s, W_s)$}
	{	
		$q_j \leftarrow  \max\{q_j,  \frac{\alpha_j^{\beta}}{{\alpha_{th}^{\beta}}}\frac{v^o_{th}}{\left|\Gamma_{th} \right|}  \left|\Gamma_j \right|\}  $\;
		$p_i \leftarrow  \max\{p_i, \frac{a_{th}}{{\lambda_{th}^{\beta}}}{\lambda_i^{\beta}}\}$\;
		$Q \leftarrow Q \cup \{q_j\}$\;	
		$P \leftarrow P \cup \{p_i\}$\;	
	}
	
	\Return{$R_s, W_s, r_{th}, w_{th}, Q, P$}
\end{algorithm}
In the TA, the platform trims $R_s$ and $W_s$ returned by the WRSA and the WWSA in such a way that both $R_s$ and $W_s$ have the same size, i.e., $\left|R_s\right|=\left|W_s\right|$, concurrently replacing either $r_{th}$ or $w_{th}$ (line 4 or 11). After the trimming process, the platform calculates temporary fee $q_j$ for each $r_j \in R_s$ and temporary payment $p_i$ to each $w_i \in W_s$. 
For the calculation of $q_j$ and $w_i$, the platform uses the ratio value $v_{th}/({\alpha_{th}^{\beta}}\left|\Gamma_{th}\right|)$ of $r_{th}$ and the ratio value $c_{th} / {\lambda_{th}^{\beta}}$ of $w_{th}$, which consequently make $q_j$ and $p_i$ the critical values to guarantee truthfulness of requesters and workers, respectively (line 14-19).
Note that we denote $q_j$ ($p_i$) temporary fee (temporary payment) because both $q_j$ and $p_i$ in the TA were initially calculated based on the assumption that all $w_i \in W_s$ will meet the deadline and all $r_j\in R_s$ will accordingly achieve their full valuation. Later when $w_i \in W_s$ completes its assigned task, both $q_j$ and $p_i$ will be updated in the Pricing Algorithm. 

After the TA, the platform checks whether the budget-balance holds (line 5). If it holds, the platform starts the matching process (line 10-12). Otherwise, the platform revokes the auction (line 6-8).
\begin{algorithm}
	\caption{Matching Algorithm (MA)}
	\label{Matching1}
	\SetKwInOut{Input}{Input}
	\SetKwInOut{Output}{Output}
	
	\Input{$R_s, W_s, r_{th}, w_{th}, \beta$}
	\Output{$Match, Q, P$}
	
	\text{Assume $\frac{v^o_1}{\alpha_1^{\beta}\left|\Gamma_1\right|}\geq \dots \geq \frac{v^o_{\left|R_s\right|}}{\alpha_{\left|R_s\right|}^{\beta} \left|\Gamma_{\left|R_s\right|}\right|}\geq \frac{v^o_{th}}{\alpha_{th}^{\beta}\left|\Gamma_{th}\right|}$}\;
	\text{Assume $\frac{c_1}{{\lambda_1}^{\beta}}\leq \dots \leq \frac{c_{\left|W_s\right|}}{{\lambda_{\left|W_s\right|}^{\beta}}} \leq \frac{c_{th}}{{\lambda_{th}^{\beta}}}$}\;
	
	$Match \leftarrow \emptyset$\;
	
	$R_s, W_s, r_{th}, w_{th}, Q, P =  \text{TA}(R_s, W_s, r_{th}, w_{th}, \beta)$\;
	\If{$\underset{p_i \in P}{\sum}p_i > \underset{q_j \in Q}{\sum}q_j$}
	{
		$R_s \leftarrow \emptyset, W_s \leftarrow \emptyset$\;	
		$Q \leftarrow \emptyset, P \leftarrow \emptyset$\;
		\Return{$Match, Q, P$}
	}	
	
	\For{$(r_j, w_i) \in (R_s, W_s)$}
	{	
		
		$Match \leftarrow Match \cup (r_j, w_i) $\;
		
	}
	\Return{$Match, Q, P$}
\end{algorithm}

\subsection{Pricing Step}
Unlike the existing works where the fee for $R_s$ and the payment to $W_s$ are determined before the task submission and do not change, our mechanism determines the final fee and payment, called \textit{effective} fee and payment, depending on the task valuation at the task submission time ($t_i^{sub}$) of each $w_i \in W_s$. As detailed in Algorithm \ref{PA}, when $w_i \in W_s$ is matched to $r_j \in R_s$ and submits its task result at $t_i^{sub}$, the platform decides the effective fee $q'_j$ for $r_j$ and the effective payment $p'_i$ to $w_i$ in proportion to the ratio of achieved valuation $v_j(t_i^{sub})$ to the original (full) valuation $v_j^o$, respectively (line 4-10). The platform can run the pricing processes for each match $(r_j, w_i)$ in parallel. Note that temporary $q_j$ and $p_i$ were calculated in the MA. 
\begin{algorithm}
	\caption{Pricing Algorithm (PA)}
	\label{PA}
	
	\SetKwInOut{Input}{Input}
	\SetKwInOut{Output}{Output}
	
	\Input{$Match, Q, P$}
	\Output{$Q', P'$}
	
	$Q' \leftarrow \emptyset, P' \leftarrow \emptyset$\; 
	$q'_j \leftarrow 0, \forall w_i \in Match$\; 	
	$p'_j \leftarrow 0, \forall r_j \in Match$\;
	
	\ForAll{\texttt{$(r_j, w_i) \in Match$}}
	{   
		$t_i^{sub}$ $\leftarrow$ \text{$w_i$'s task submission time}\;
		$q_j \leftarrow$ \text{‘fee for $r_j$’ $\in Q$ calculated before $t_{sub}^i$}\;
		$p_i \leftarrow$ \text{‘payment to $w_i$’ $\in P$ calculated before $t_{sub}^i$}\; 
		
		$q'_j \leftarrow \frac{v_j(t_i^{sub})}{v_j^o}q_j, Q' \leftarrow Q'\cup \{q'_j\}$\;
		$p'_i \leftarrow \frac{v_j(t_i^{sub})}{v_j^o}p_i, P' \leftarrow P'\cup \{p'_j\}$\;
	}
	\Return{$Q', P'$}
\end{algorithm}

The entire process of H-ESWM is shown in Algorithm \ref{ESWM}. The platform first runs the WRSA and the WWSA to decide the winners in a double auction. Then, it runs the MA to match $R_s$ to $W_s$. Lastly, the platform decides both effective fee and payment in the PA.	
\begin{algorithm}
	\caption{H-ESWM}
	\label{ESWM}
	
	\SetKwInOut{Input}{Input}
	\SetKwInOut{Output}{Output}
	
	\Input{$R, W, K, \beta$}
	\Output{$Match, P', Q'$}
	
	$R_s, r_{th}$ = WRSA($R,K$)\;
	$W_s, w_{th}$ = WWSA($W,K$)\;
	
	$Match, Q, P$ = MA($R_s, W_s, r_{th}, w_{th}$)\;
	$Q', P'$ = PA($Match, Q, P $)\;	
	\Return{$Match, Q', P'$}
\end{algorithm}

\section{Performance Evaluation}
\label{part:5}
In this section, we evaluate the performance of the ESWM mechanism in both short-term and long-term scenarios. As crowdsourcing service platforms may \textit{continuously} compete with \textit{each other} in real IoT-based crowdsourcing systems, the current performance of platforms can also affect the performance in the future competitions. Thus, we first compare the performance of proposed ESWM mechanism with one of the existing works \cite{team1} in the short-term scenario with one round of auction. Then, we extend the evaluation to the long-term scenario where the previous performance metrics (average requester utility or worker utility) affect the current recruitment of requesters and workers. We let \cite{team1} and ESWM mechanism compete with each other and see how effectively they can attract participants. 

\subsection{Performance Metrics}
As for performance metrics, we consider the social welfare, the platform utility, the average requester utility, and the average worker utility. We compare the performance metrics of our mechanism to those of the \textit{benchmark} \cite{team1} whose winner selection process is also based on the greedy algorithm, but only considering the ratio of $v_j^{max}/|\Gamma_j|$ and $c_i$. Lastly, we prove that our mechanism achieves the four desirable economic properties. Before we provide the simulation results, we define each performance metric as below. 

\subsubsection{Social Welfare}
The social welfare is divided into two categories: 1) the Na\"\i ve social welfare (NSW) as defined in (\ref{eq:20}), and 2) the expected social welfare (ESW) as defined in (\ref{eq:4}).
While the expected social welfare takes potential task depreciation into account, the Na\"\i ve social welfare merely assumes the ideal case where all the tasks are completed in time. 

\subsubsection{Platform Utility}
The platform utility is as defined in (\ref{eq:3}). 

\subsubsection{Average Requester \& Worker Utility}
The average requester utility is defined as follows
\begin{equation}
\bar{u}_r=\frac{\sum_{r_j \in R} u_j^r}{\left|R\right|}.
\end{equation}
Similarly, the average worker utility is defined as follows
\begin{equation}
\bar{u}_p=\frac{\sum_{p_i \in W} u_i^p}{\left|W\right|}.
\end{equation}

\begin{figure*}[!htb]
	\centering
	\subfloat[Social Welfare \label{subfig-2:sw}]{%
		\resizebox{0.24\textwidth}{!}{\includegraphics[width=\linewidth]{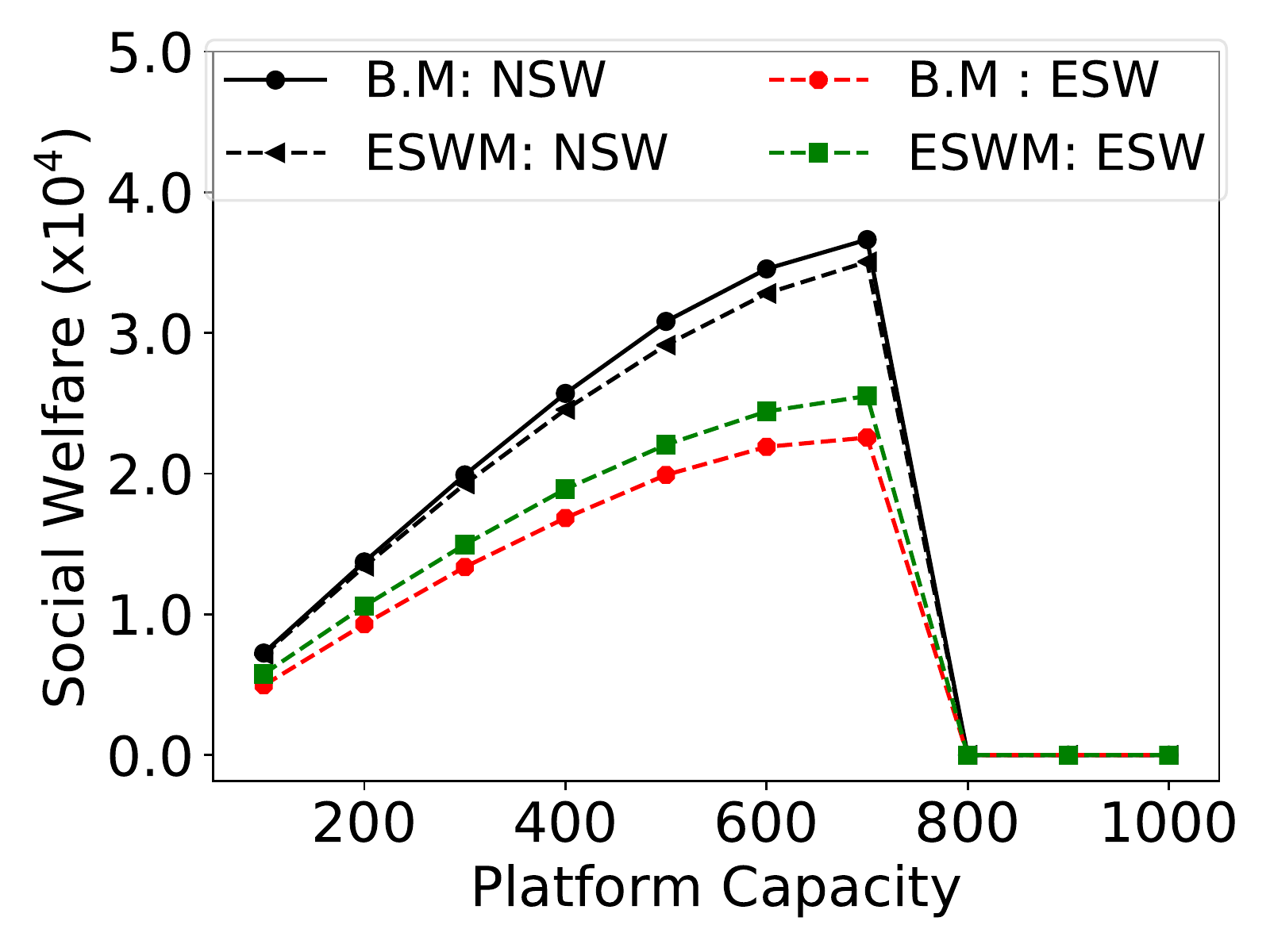}}}
	\subfloat[Platform Utility \label{subfig-2:platform}]{%
		\resizebox{0.24\textwidth}{!}{\includegraphics[width=\linewidth]{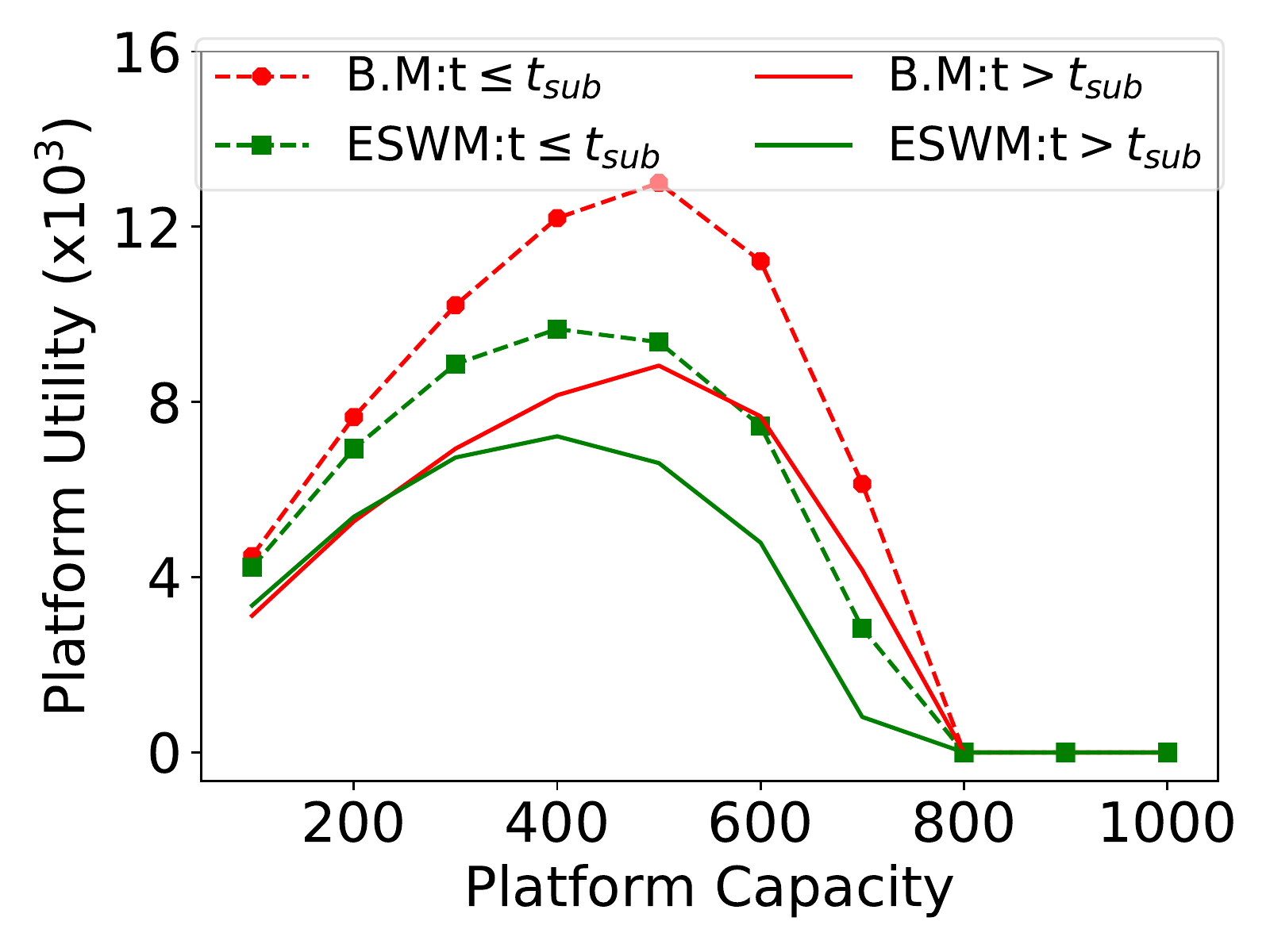}}}
	\hfill
	\subfloat[Average Requester Utility \label{subfig-2:requester}]{%
		\resizebox{0.24\textwidth}{!}{\includegraphics[width=\linewidth]{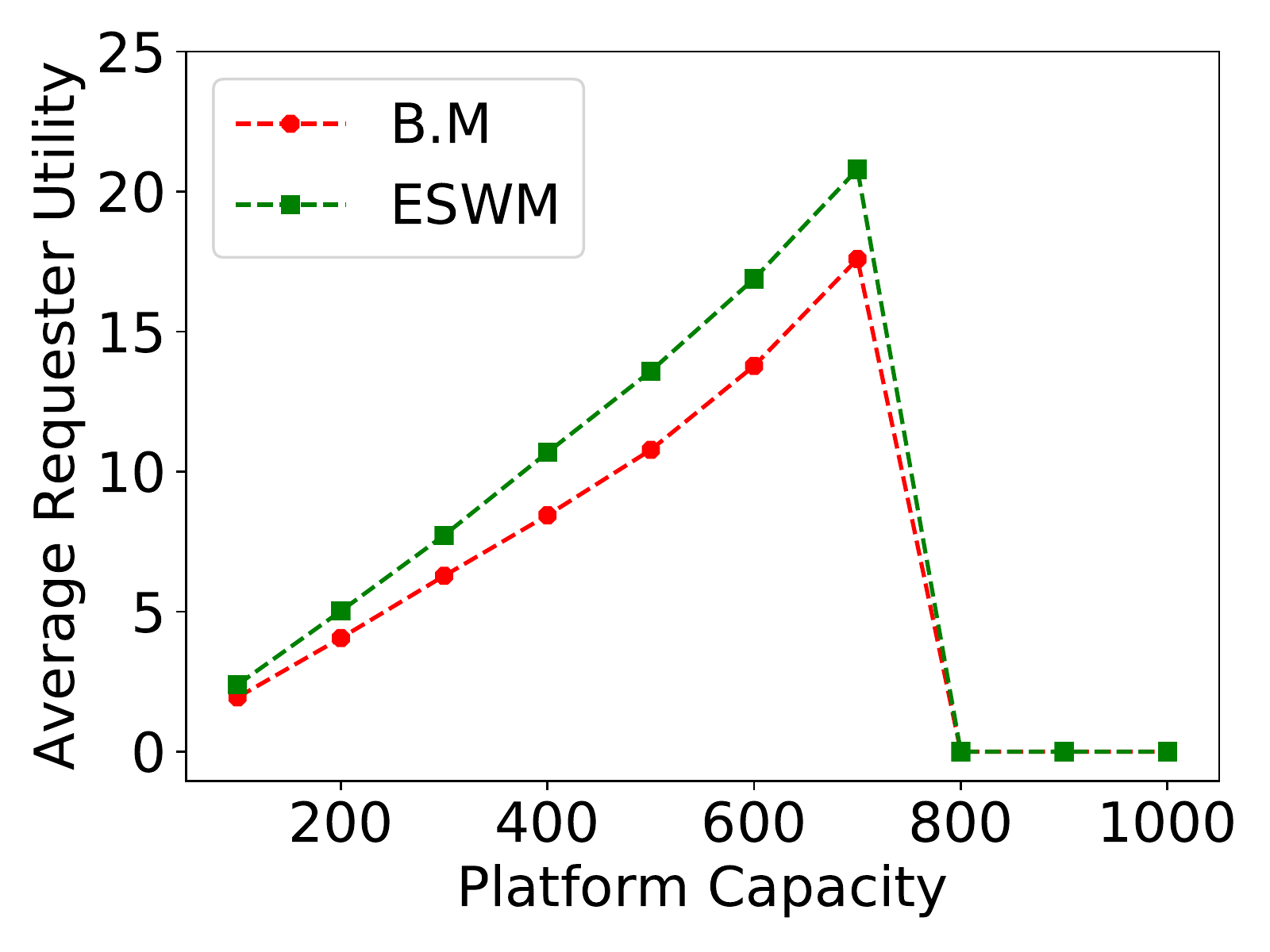}}}
	\subfloat[Average Worker Utility \label{subfig-2:worker}]{%
		\resizebox{0.24\textwidth}{!}{\includegraphics[width=\linewidth]{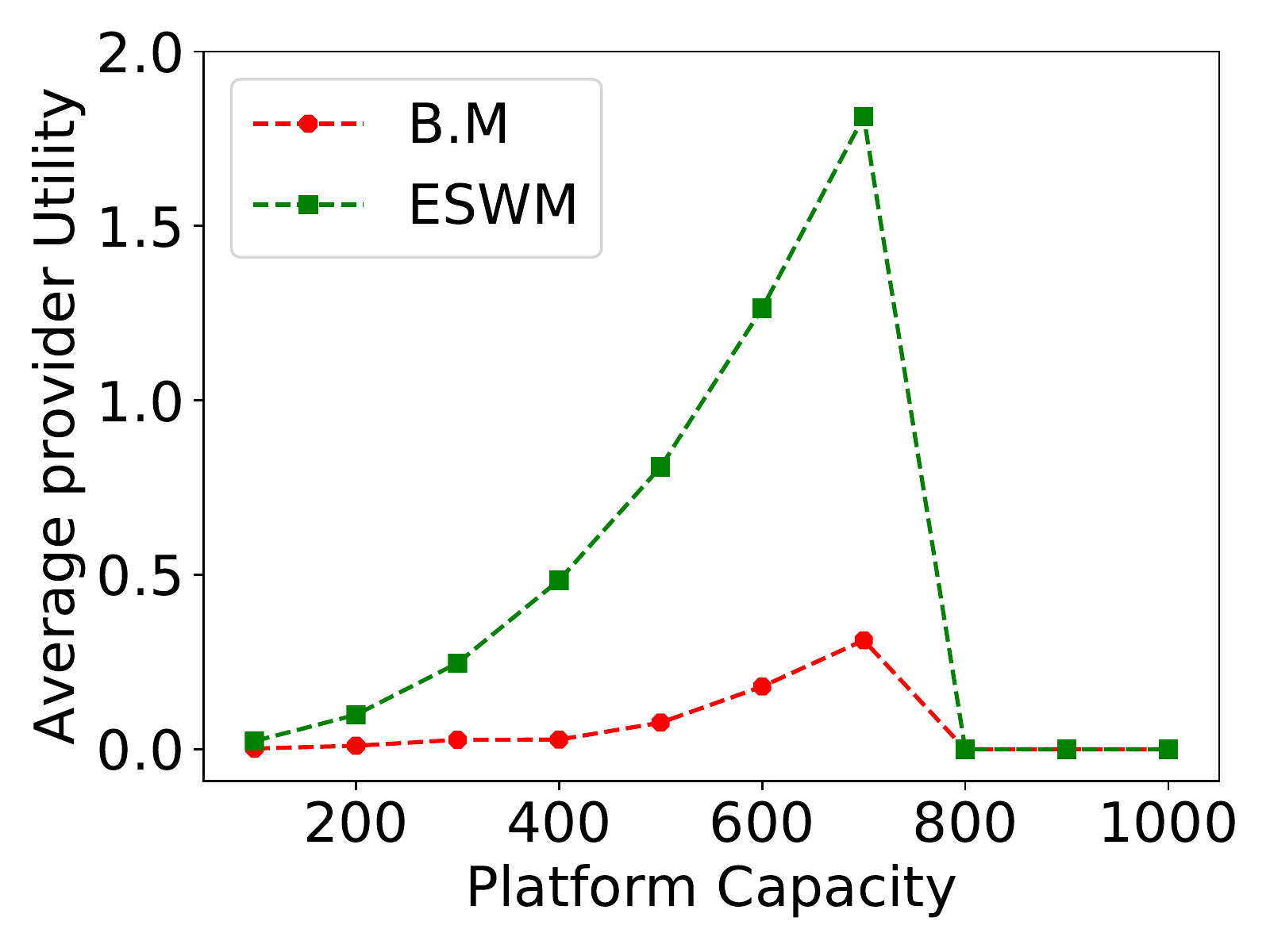}}}
	\caption{Benchmark (B.M) vs ESWM in a Single Auction} 
	\label{fig:5}
\end{figure*}

\subsection{Simulation Setting}
In our simulations, we uniformly distribute $v_j^{max}$, $\left|\Gamma_j\right|$, $t_j^d$, and $t_j^{ex}$ for requesters over (0, 100] and [1, 10], (0,100], and [$t_j^d$, 1.5$t_j^d$], respectively. And we uniformly distribute $\alpha_j$ over (0, 100] to include the case where a task becomes valueless right after the deadline (when $v_j^{max}=100$ and $\alpha_j=100$). 
For workers, we uniformly distribute $c_i$ and $\mu_i$ over (0, 10] and (0, 1.5], respectively. All the simulation results for the performance metrics are averaged over 200 runs.     

\subsection{Benchmark vs ESWM in a Single Auction}
\label{sim:1}
In this simulation, we compare the performance of the ESWM mechanism to the benchmark when both are given the same 1,000 requesters and 2,000 workers with $\beta=0.5$.  

Fig. \ref{subfig-2:sw} shows the social welfare of the benchmark and the ESWM mechanism, increasing the platform capacity from 100 to 1000 in an increment of 100. In terms of the na\"\i ve social welfare, both mechanisms achieve almost the same value. However, the ESWM mechanism produces higher expected social welfare than the benchmark. This can be attributed to the fact that the ESWM mechanism focuses on the estimated valuation which will be achieved in the submission time rather than the original valuation which may depreciate. In contrast, the benchmark mechanism shows a wider gap between its na\"\i ve social welfare and expected one, failing to capture the potential task depreciation.

Fig. \ref{subfig-2:platform} shows the platform utility achieved by the benchmark and the ESWM mechanism, before and after the task submission of workers. In both cases, the benchmark mechanism produces higher platform utility than the ESWM mechanism. Such difference results from the difference in the pricing step. In the WRSA and the WWSA, the fee for $r_j \in R_s$ and the payment to $w_i \in W_s$ were calculated as 
\begin{equation}
\label{eq:23}
q_j = \frac{{\alpha_j^{\beta}}}{{\alpha_{th}^{\beta}}}\frac{v^{max}_{th}}{\left|\Gamma_{th}\right|}\left|\Gamma_j\right|, 
p_i = \frac{{\lambda_i^{\beta}}}{{\lambda_{th}^{\beta}}}c_{th}.
\end{equation}
For $r_j \in R_s$, $\alpha_j$ tends to be smaller than $\alpha_{th}$ since the ESWM mechanism iteratively selects winners with the maximum $v_j^{max} / (\alpha_j\left|\Gamma_j\right|)$. Consequently, $q_j$ is smaller than that of the benchmark calculated as $v^{max}_{th}\left|\Gamma_{j}\right|/\left|\Gamma_{th}\right|$. Similarly, $\lambda_i$ for $w_i \in W_s$ is likely to be larger than $\lambda_{th}$, as the ESWM mechanism iteratively selects as a winner worker with the minimum $c_i / \lambda_i^{\beta}$. Consequently, $p_i$ is larger than that of the benchmark calculated as $c_{th}$. However, after the task submission, the difference significantly decreases due to more frequent unpunctuality in the benchmark, which consequently inflicts more utility loss. 

Fig. \ref{subfig-2:requester} and Fig. \ref{subfig-2:worker} show the average requester utility and the average worker utility, respectively. As mentioned in the analysis of Fig. \ref{subfig-2:platform}, the ESWM mechanism charges $R_s$ less and rewards $W_s$ more. Consequently, the ESWM mechanism makes both requesters and workers achieve higher average utility by sacrificing its own utility. In a myopic strategy, it is not a rational decision for a platform to sacrifice its own utility for participants. However, as demonstrated in the next subsection, such sacrifice can bring about a positive effect to the platform itself in the long-term strategy. 
\subsection{Benchmark vs ESWM with Re-selection}
\label{bench_ESWM}
\begin{figure*}[!htb]
	\centering
	\subfloat[Social Welfare \label{subfig-1:sw}]{%
		\resizebox{0.24\textwidth}{!}{\includegraphics[width=\linewidth]{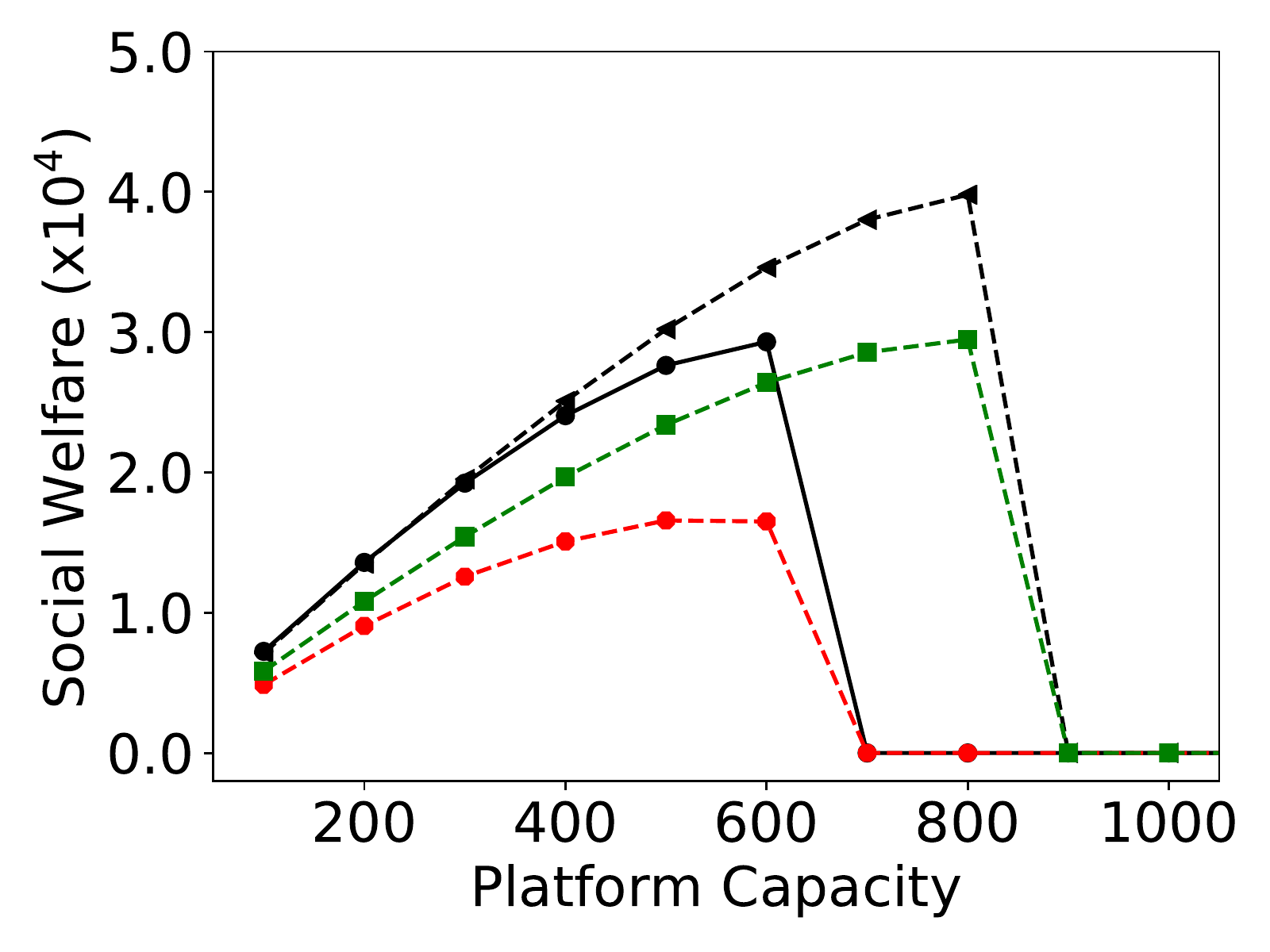}}}
	\subfloat[Platform Utility \label{subfig-1:platform}]{%
		\resizebox{0.24\textwidth}{!}{\includegraphics[width=\linewidth]{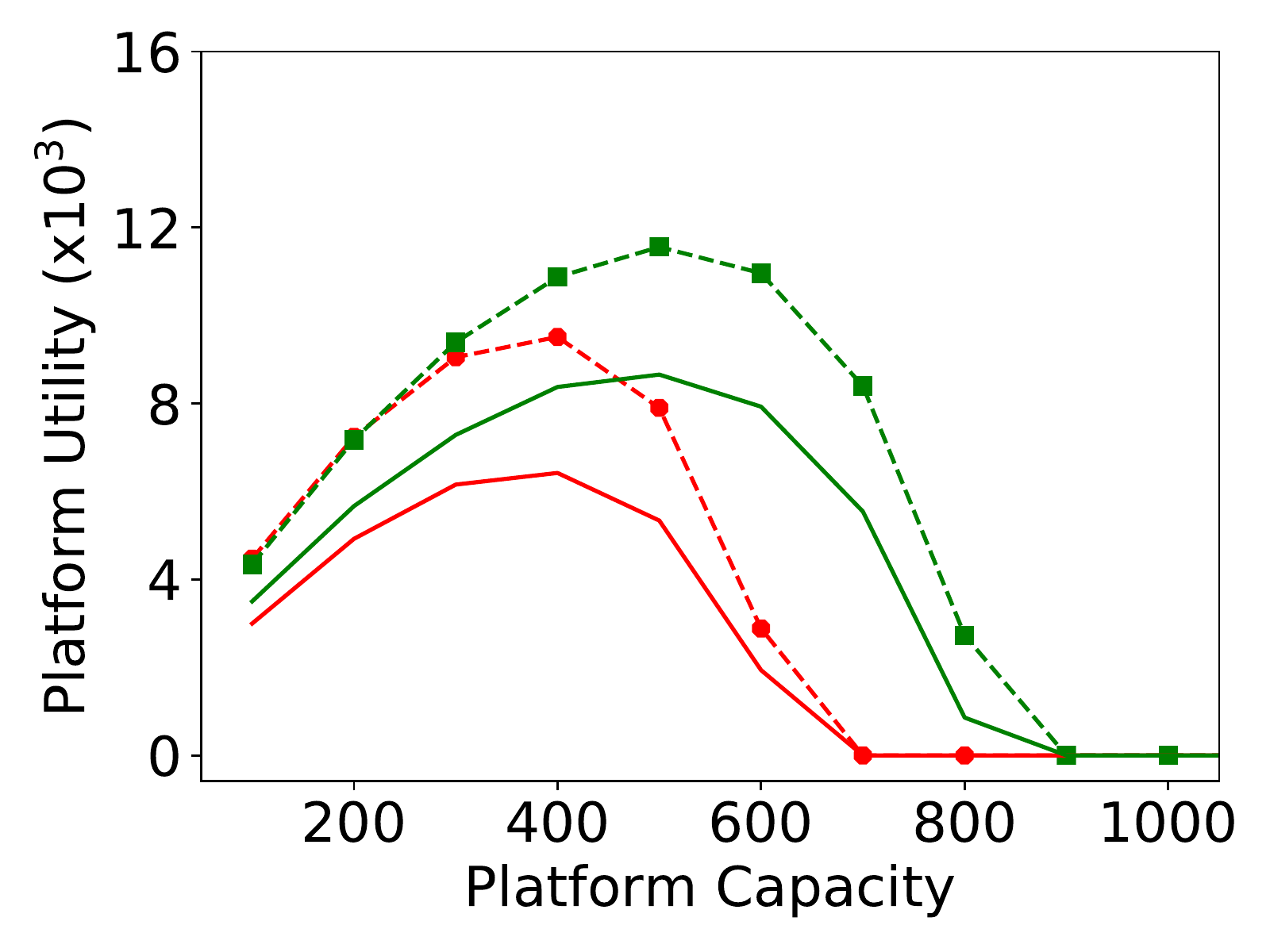}}}
	\hfill
	\subfloat[Average Requester Utility \label{subfig-1:requester}]{%
		\resizebox{0.24\textwidth}{!}{\includegraphics[width=\linewidth]{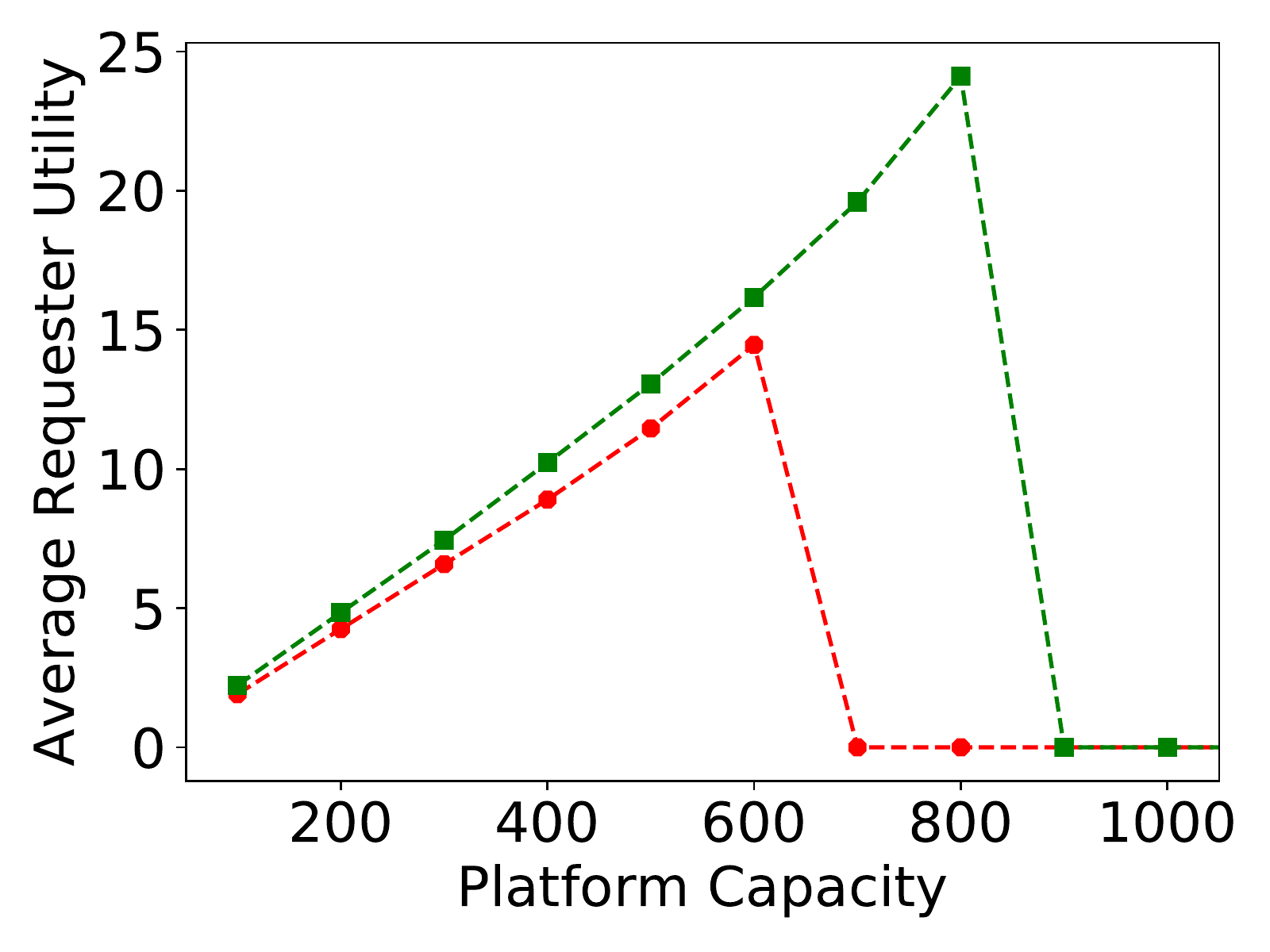}}}
	\subfloat[Average Worker Utility \label{subfig-1:worker}]{%
		\resizebox{0.24\textwidth}{!}{\includegraphics[width=\linewidth]{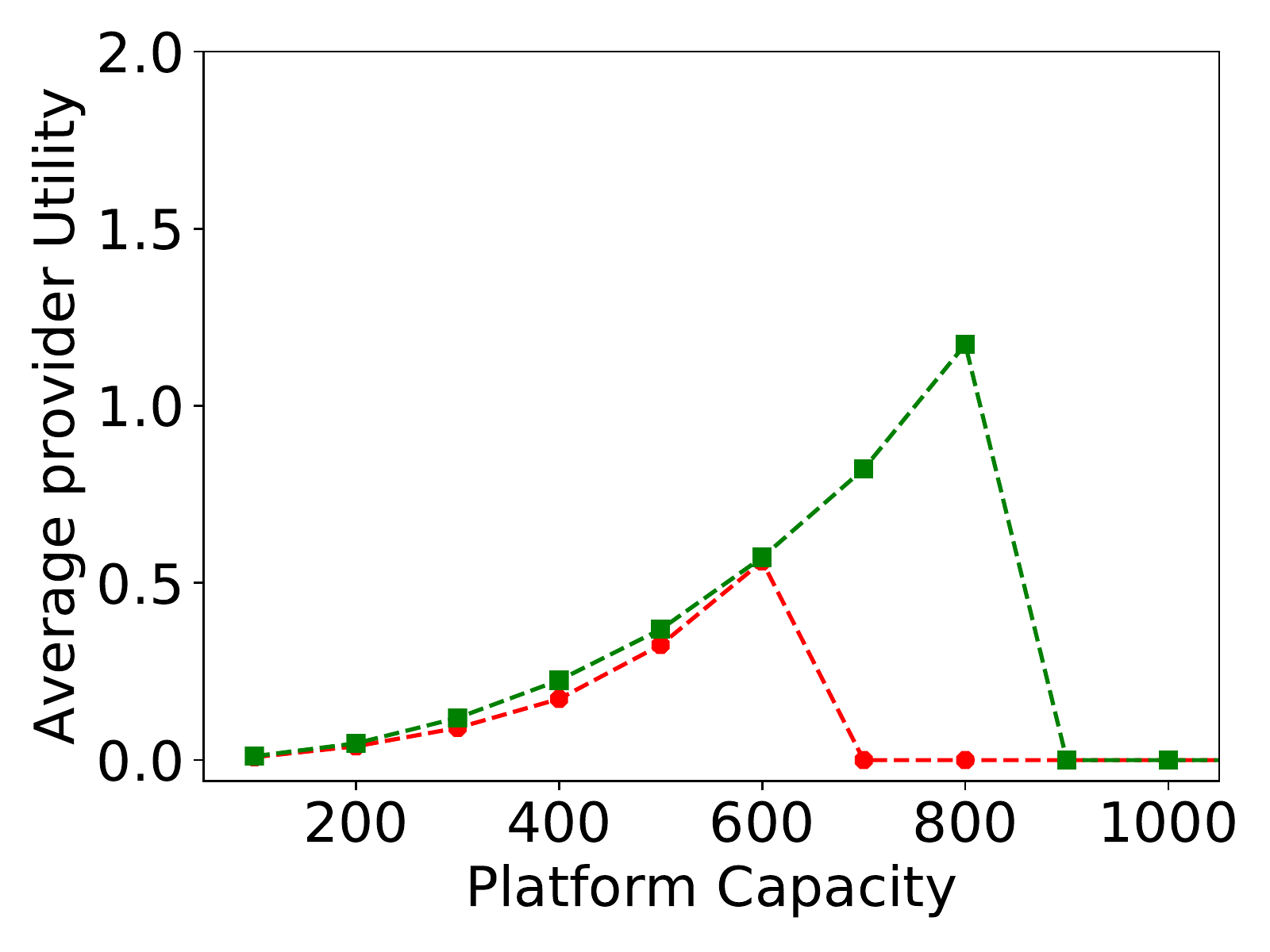}}}
	\caption{Benchmark (B.M) vs ESWM with Re-selection}
	\label{fig:6}
\end{figure*}
In the previous simulation, we showed that the ESWM mechanism achieves higher social welfare, average requester utility, and average worker utility than the benchmark, at the cost of its own utility. As an initial stage of the competition, both mechanisms were given the same number of requesters and workers in the previous simulation. That is, we assumed that both mechanisms could attract the same number of requesters and workers. 

However, in reality, how attractive each mechanism is to participants can differ since each provides different utility for participants. As all participants are rationally selfish, they select the mechanism that would provide higher utility to them. Consequently, when the benchmark and our mechanism compete in a crowdsourcing system, the number of participants that each mechanism attracts can vary depending on the average utility each provides.
Thus, to reflect such different attractiveness of each mechanism to participants, we define participation probability functions for requesters and workers as   
\begin{equation}
\label{eq:21}
p^r(\bar{u}_{r^A}, \bar{u}_{r^B}) = (\frac{\sqrt{\bar{u}_{r^A}}}{\sqrt{\bar{u}_{r^A}} + \sqrt{\bar{u}_{r^B}}}, \frac{\sqrt{\bar{u}_{r^B}}}{\sqrt{\bar{u}_{r^A}} + \sqrt{\bar{u}_{r^B}}}),
\end{equation}
\begin{equation}
\label{eq:22}
p^p(\bar{u}_{p^A}, \bar{u}_{p^B}) = (\frac{\sqrt{\bar{u}_{p^A}}}{\sqrt{\bar{u}_{p^A}} + \sqrt{\bar{u}_{p^B}}}, \frac{\sqrt{\bar{u}_{p^B}}}{\sqrt{\bar{u}_{p^A}} + \sqrt{\bar{u}_{p^B}}}),
\end{equation}
where $\bar{u}_{r^A}$ ($\bar{u}_{r^B}$) and $\bar{u}_{p^A}$ ($\bar{u}_{p^B}$) denote the average utilities of requesters and workers who joined mechanism A (mechanism B) in the previous recruitment, respectively. Given the average utility of participants in each mechanism, the participation probability function returns a tuple of probabilities that a participant decides to join each mechanism. In this work, we set the participation probability proportional to the square root of the average utility of participants obtained from the previous simulation results, based on \cite{concave}. According to the reference, the concavity of the square root function captures the diminishing impact of the utility on the participation probability. 

Based on the probabilities, each participant decides which mechanism it will join in the current recruitment. Consequently, when there exist two different mechanisms competing in a crowdsourcing system, a mechanism which has provided higher utility to participants can expect to attract more participants. We call such decision-making process of participants \textit{re-selection}. 
In the re-selection simulation, the benchmark mechanism and the ESWM mechanism compete with each other in a system with 2,000 requesters, 4,000 workers, and $\beta=0.5$. 

Fig. \ref{subfig-1:sw} shows the simulation results of the social welfare under such re-selection scenario. Compared to the previous simulation results, the ESWM mechanism achieves much higher na\"\i ve and expected social welfare than the benchmark. As the ESWM mechanism provided both requesters and workers with higher average utility, it is more appealing to both requesters and workers than the benchmark, which attracts more of them. Consequently, such quantitative growth leads to higher chance of taking beneficial participants, which ultimately increases the social welfare of the ESWM mechanism. Note that in the point of view of the platform, beneficial participants means a requester with high $v_j^{max}$ and low $\alpha_j$, and a worker with low $c_i$ and high $\lambda_i$. In addition, even with more number of participants, especially workers, the ESWM mechanism can still handle the task requests better than the benchmark (800 task requests while the benchmark can handle up to 600). The reason for such difference can be inferred from the platform utility. As shown in Fig. \ref{subfig-1:platform}, the platform utility of the benchmark is 0 after $K=900$. Thus, the platform, which is rational and selfish, revokes the double auction as its budget-balance condition is not satisfied.  

Moreover, unlike the previous simulation result, Fig. \ref{subfig-1:platform} shows that the ESWM mechanism achieves higher platform utility than the benchmark. For the same reason in the social welfare, such increased platform utility is achieved as the ESWM mechanism attracts more participants, which enables the platform to get better behaving participants. As a result, even though the ESWM mechanism sacrificed its own utility in the initial stage, its utility loss is compensated by attracting more participants in the long-term competition. In addition, the platform capacity to achieve the maximum platform utility increases from 400 to 500, which results in the higher maximum platform utility.

Fig. \ref{subfig-1:requester} and Fig. \ref{subfig-1:worker} show the average requester utility and the average worker utility, respectively. In both results, the ESWM mechanism and the benchmark achieve almost the same average utility as long as both can handle task requests. This can increase the number of requesters and workers participating in the ESWM mechanism, which ironically reduces the average utility of participants, especially workers. In other words, as the number of participants who join the ESMW mechanism increases, more participants fail to be selected as winners, which results in more zero utilities. As a result, the average participant utility of the ESWM mechanism decreases. Based on our observation from Fig. \ref{subfig-1:requester} and Fig. \ref{subfig-1:worker}, we can anticipate that there will not be a significant number of re-selections since the benchmark and the ESWM mechanism offer similar average utilities for both requesters and workers. That is, the benchmark and the ESWM mechanism have approached near the \textit{balance point} where both mechanisms are equally attractive to participants. Specifically, we can define the balance point as the case where $p^r(\bar{u}_{r^A}, \bar{u}_{r^B})$ and $p^p(\bar{u}_{r^A}, \bar{u}_{r^B})$ are discrete uniform distributions. 

\subsection{Effect of $\beta$}
\begin{figure*}[!htb]
	\centering
	\subfloat[Platform Utility \label{subfig-wide:platform}]{%
		\resizebox{0.23\textwidth}{!}{\includegraphics[width=\linewidth]{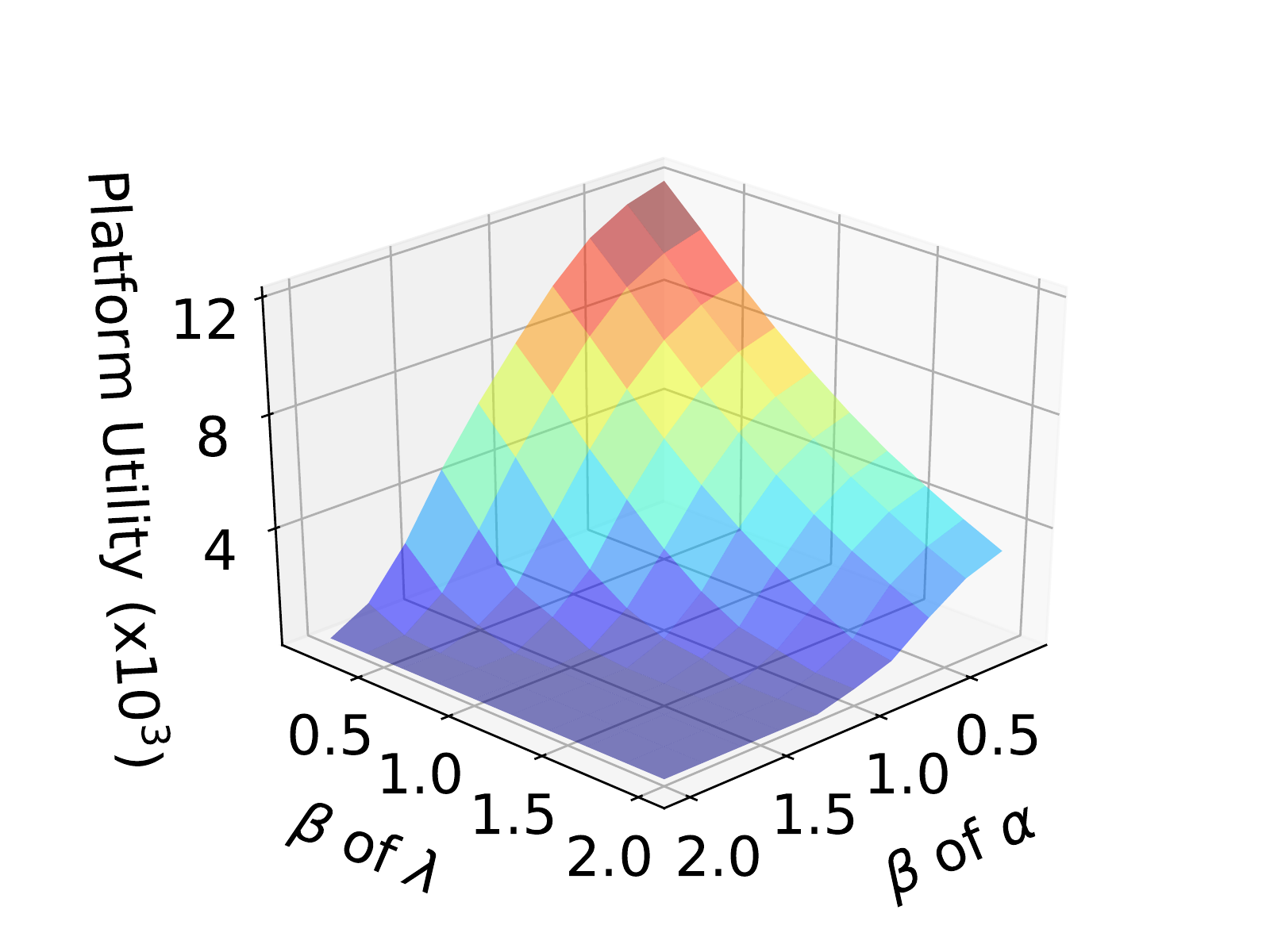}}}
	\subfloat[Average Requester Utility \label{subfig-wide:requester}]{%
		\resizebox{0.23\textwidth}{!}{\includegraphics[width=\linewidth]{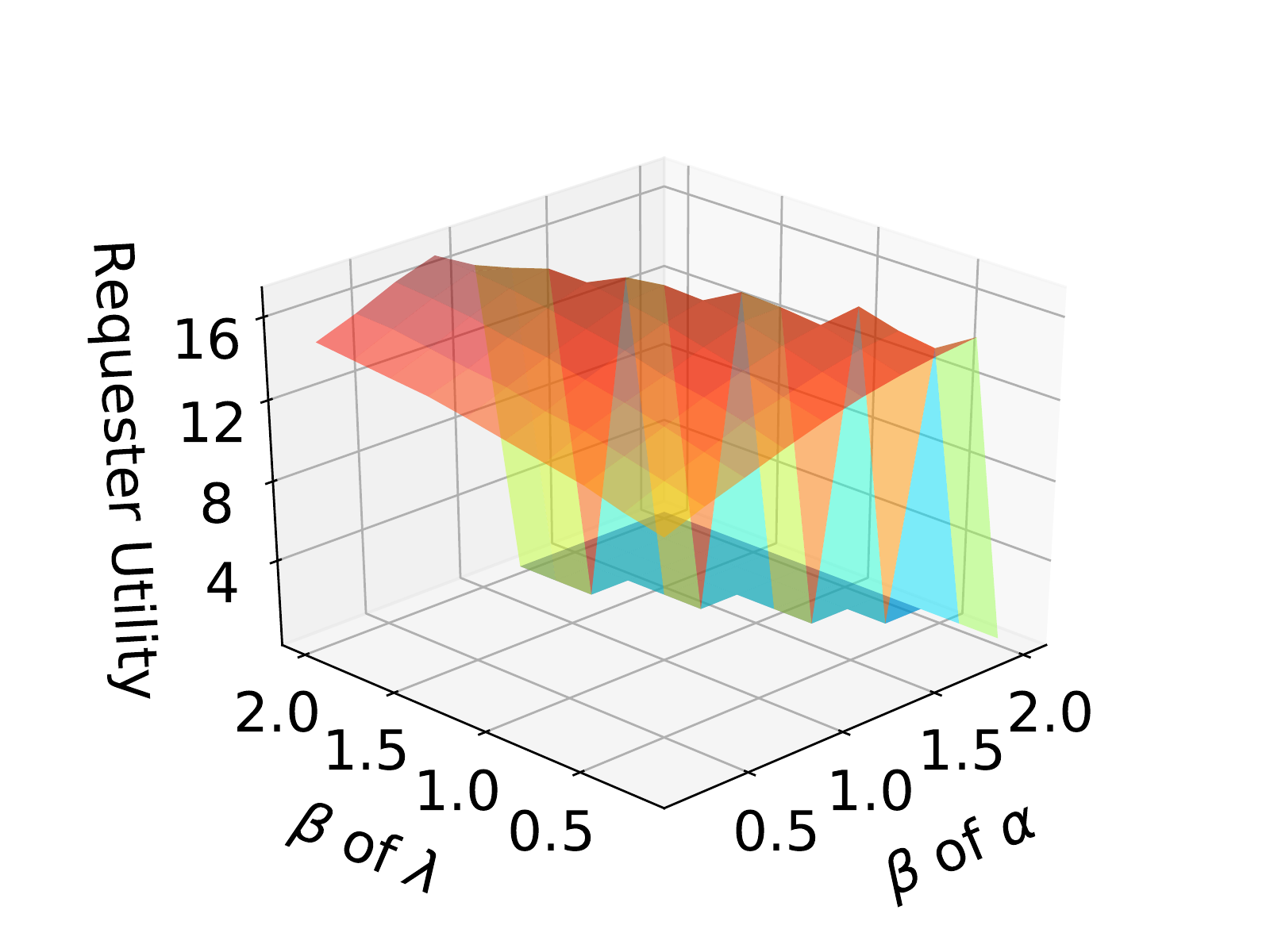}}}
	\hfill
	\subfloat[Average Worker Utility \label{subfig-wide:worker}]{%
		\resizebox{0.23\textwidth}{!}{\includegraphics[width=\linewidth]{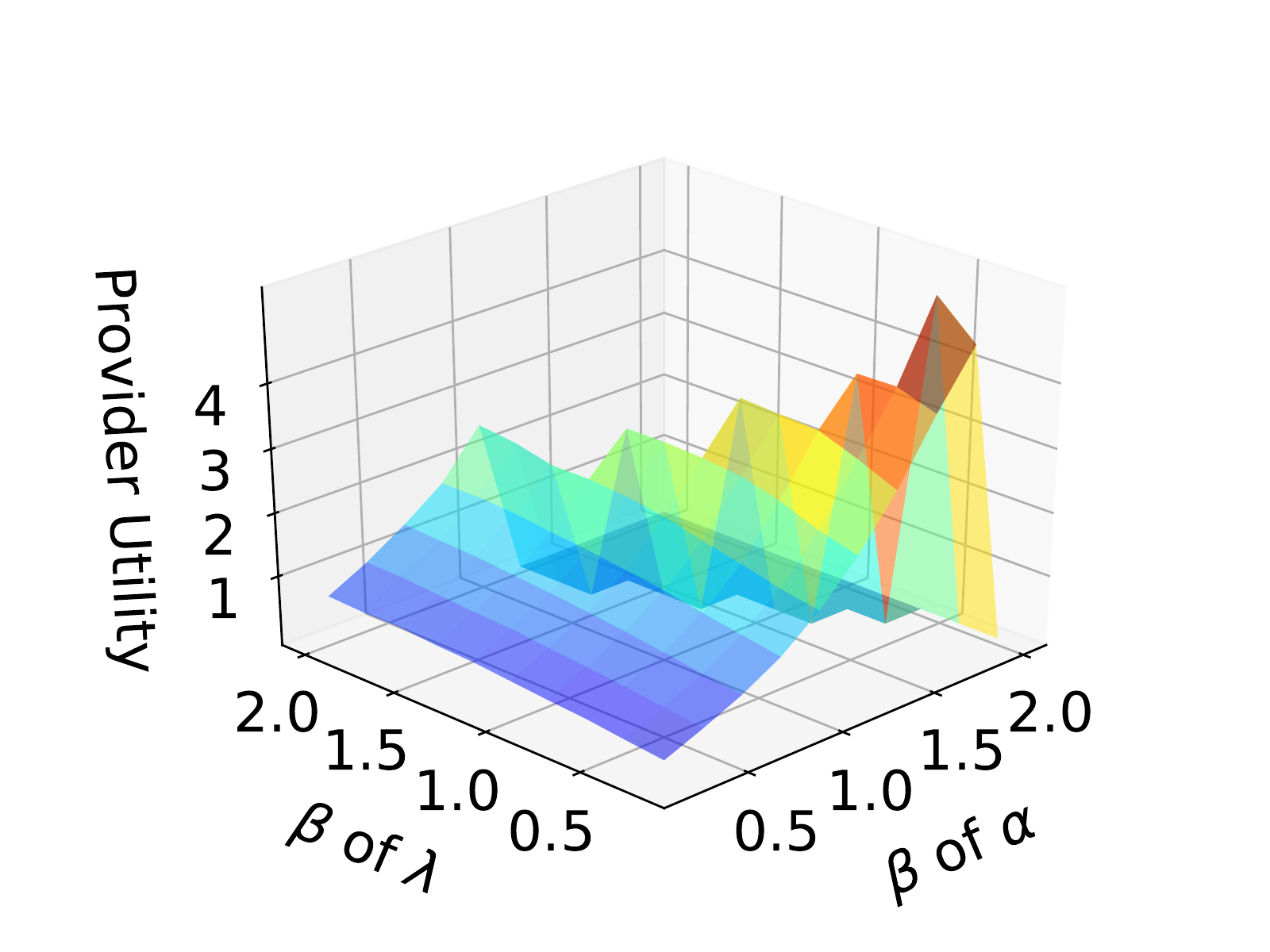}}}
	\subfloat[Social Welfare \label{subfig-wide:sw}]{%
		\resizebox{0.23\textwidth}{!}{\includegraphics[width=\linewidth]{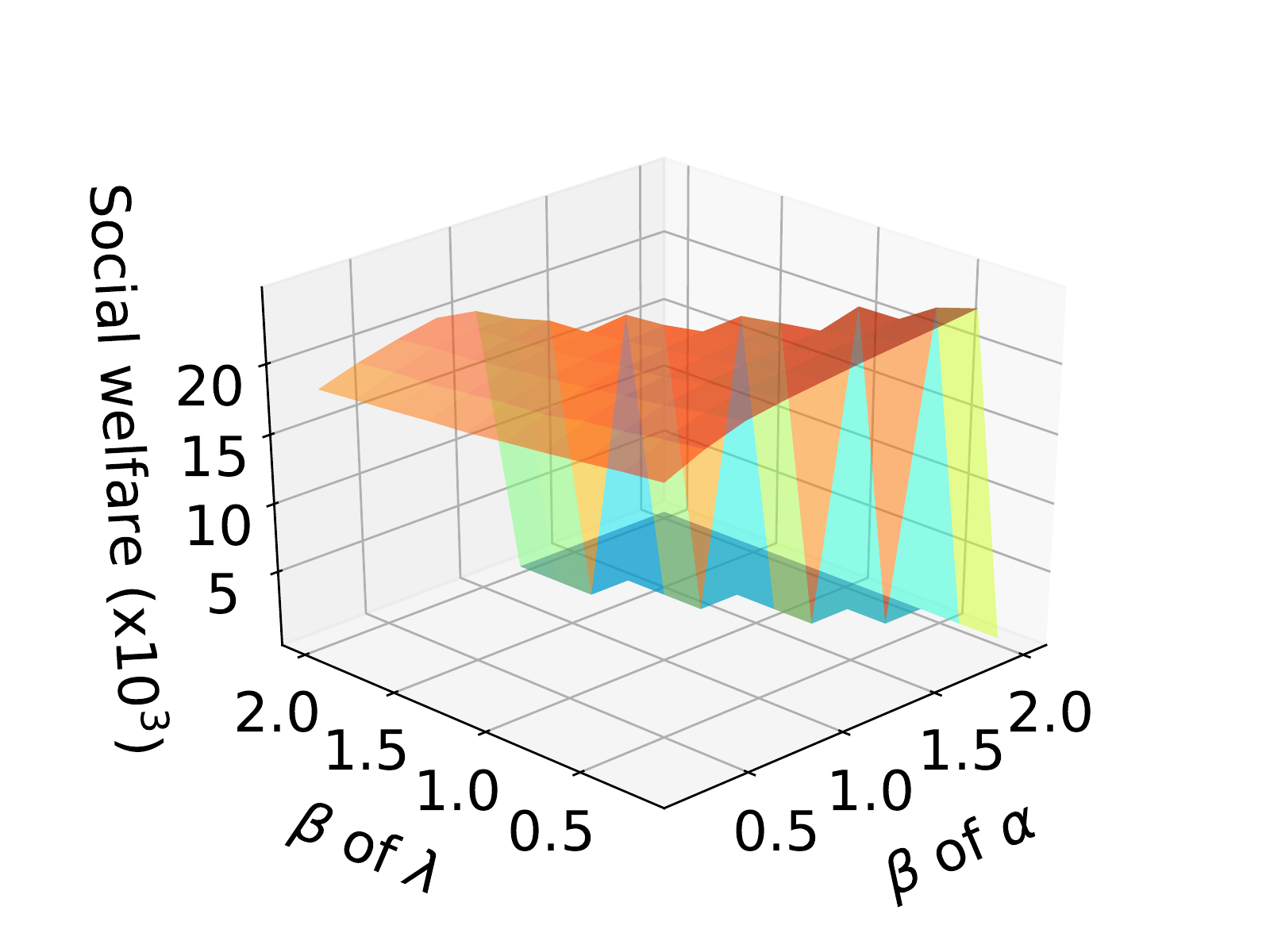}}}
	\caption{Effect of $\beta$ on Performance Metrics}
	\label{fig:9}
\end{figure*}
\begin{figure*}[!htb]
	\centering
	\subfloat[Platform Utility \label{subfig-2:platform}]{%
		\resizebox{0.23\textwidth}{!}{\includegraphics[width=\linewidth]{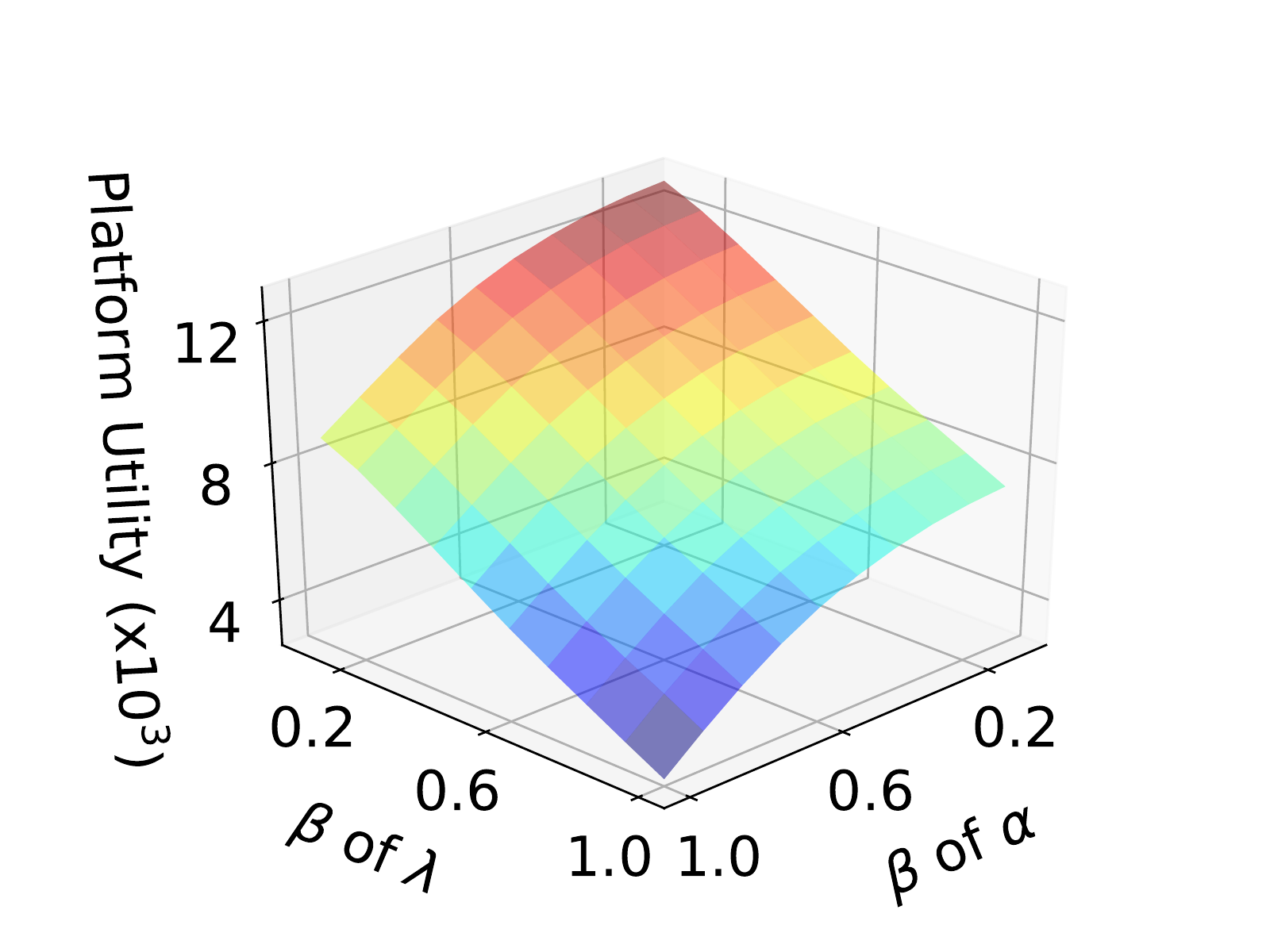}}}
	\subfloat[Average Requester Utility \label{subfig-2:requester}]{%
		\resizebox{0.23\textwidth}{!}{\includegraphics[width=\linewidth]{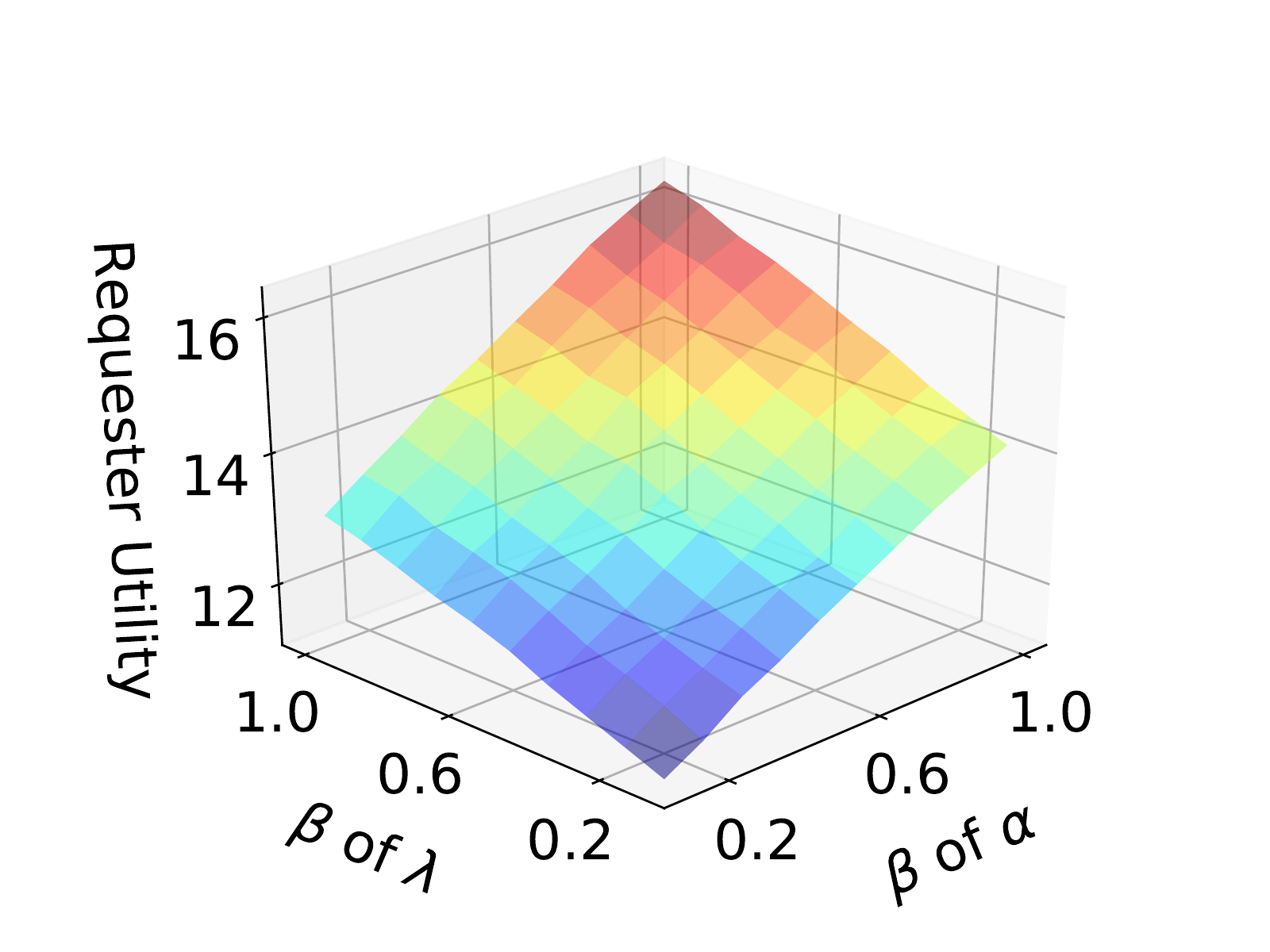}}}
	\hfill
	\subfloat[Average Worker Utility \label{subfig-2:worker}]{%
		\resizebox{0.23\textwidth}{!}{\includegraphics[width=\linewidth]{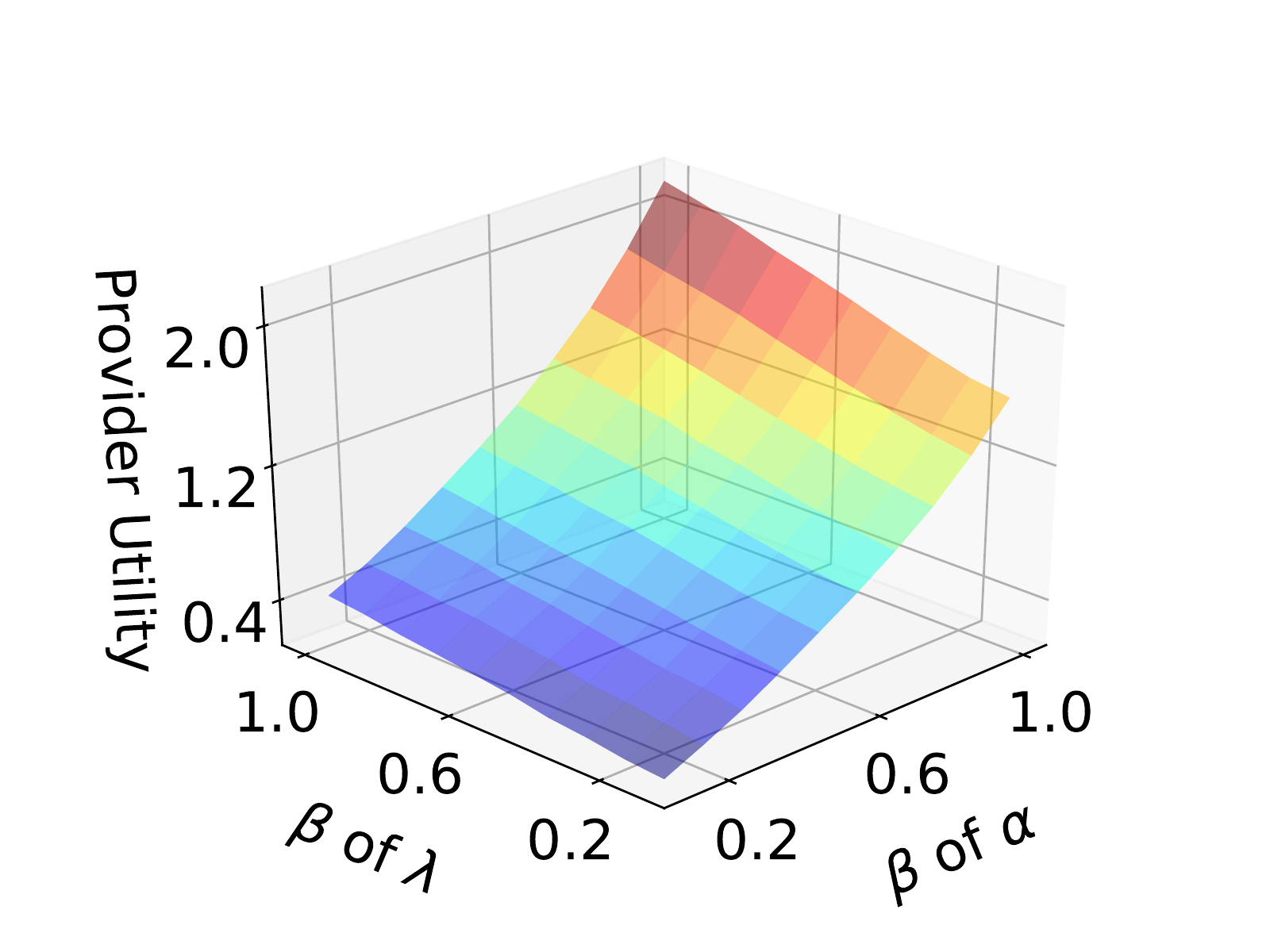}}}
	\subfloat[Social Welfare \label{subfig-2:sw}]{%
		\resizebox{0.23\textwidth}{!}{\includegraphics[width=\linewidth]{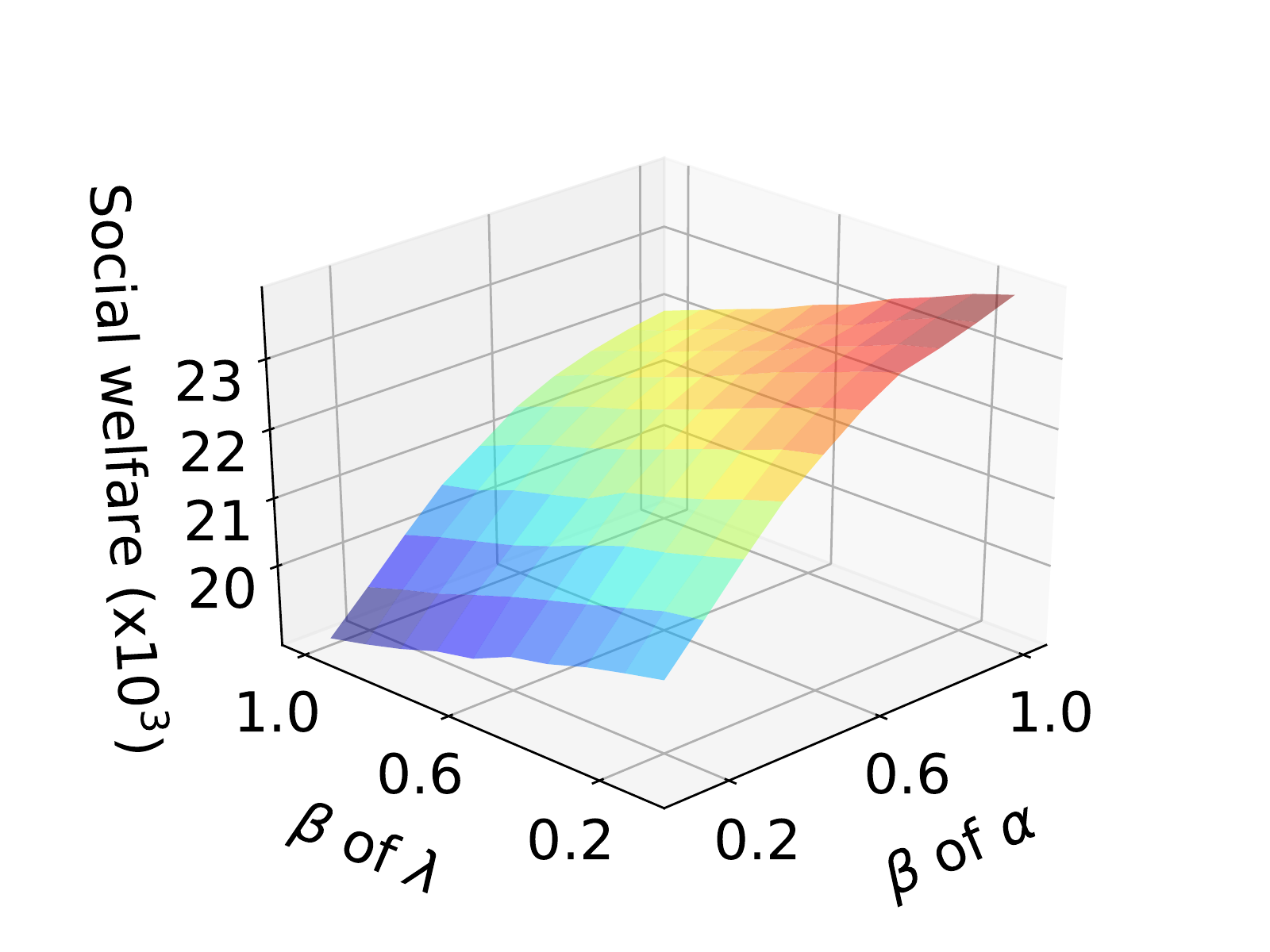}}}
	\caption{Zoomed-in Figures of Effect of $\beta$ on Performance Metrics}
	\label{fig:8}
\end{figure*}
In the previous simulations, we fixed $\beta$ which decides the weights on $\alpha_j$ for $r_j$ and $\lambda_i$ for $w_i$ in the winner selection step and the pricing step. To analyze the effect of $\beta$, we evaluate the performance metrics by vary $\beta$ over (0, 2] in an increment of 0.1. In this simulation, we set $\left|R\right|$ and $\left|W\right|$ to 1,000 and 2,000, respectively, with $K=500$.  

Fig. \ref{fig:9} shows how $\beta$ affects the four performance metrics. In Fig. \ref{subfig-wide:platform}, the ESWM mechanism achieves higher platform utility as it sets lower $\beta$ on both $\alpha$ and $\lambda$. Considering the effect of $\beta$ on $q_j$ and $p_i$ in the pricing process, such phenomenon is easily comprehensible. In the pricing process, as $\beta$ increases, the fee charged to $r_j \in R_s$ whose $\alpha_j$ is smaller than $\alpha_{th}$ decreases. Consequently, the platform utility decreases since the fees for $R_s$ which is the revenue for the platform decreases. In contrast, the reward for $w_i \in W_s$ whose $\lambda_i$ is larger than $\lambda_{th}$ increases. Since the rewards given to $W_s$ are an expenditure of the platform, its utility decreases. Interestingly, in Fig. \ref{fig:9}, when $\beta$ for both $\alpha$ and $lambda$ are set beyond 1.0, the platform utility becomes close to zero, which indicates that the platform cannot sustain its crowdsourcing service. Thus, even though appropriate $\beta$ values enable the platform to attract more participants by offering them higher average utility, a too large $\beta$ can obstruct the sustainability of the platform.  

For the same reasons in the platform utility, the average requester and worker utilities increase as $\beta$ increases as shown in Fig. \ref{subfig-wide:requester} and Fig. \ref{subfig-wide:worker}, respectively. When $v_j^{max}$ for $r_j \in R_S$ is constant, lower fees for $R_s$ lead to higher utility. With $c_i$ for $w_i \in W_s$ constant, higher reward for $W_s$ results in higher utility. As in the case of the platform utility, we can observe zero utility for both requesters and workers.       

Unlike the before-mentioned performance metrics, the expected social welfare shows a different trend. As shown in Fig. \ref{subfig-wide:sw}, while the social welfare generally increases as $\beta$ of $\alpha$ increases, it decreases as $\beta$ of $\lambda$ increases. Such inconsistency with the other performance metrics is due to the difference in the degree of importance of $\beta$ in the winning requester selection and winning worker selection processes. Compared to the role of $\beta$ of $\alpha$ in the winning requester selection, $\beta$ of $\lambda$ plays a less critical role to maximize the expected social welfare in the winning worker selection. In other words, the speed of task depreciation affects more to the expected social welfare than the workers' punctuality. In addition, setting $\beta$ of $\lambda$ large in the winning worker selection results in prioritizing workers' punctuality over their cost, which may lead to the decrease of the expected social welfare. However, considering the continuous competition with the other platforms, setting $\beta$ of $\lambda$ too small will ultimately result in the loss of both the expected social welfare and the platform utility, since $\beta$ of $\lambda$ also affect the reward for workers, which plays a critical role in attracting and retaining more workers. Thus, dynamically setting appropriate values for $\beta$ of $\alpha$ and $\beta$ is critical to the performance in the continuous competition.    

In Fig. \ref{fig:8}, we show the zoomed-in views of Fig.\ref{fig:9}, limiting $\beta$ over (0, 1] in an increment of 0.1. In each subfigure, we can more clearly observe the effect of $\beta$. 

\subsection{Proof of Desirable Economic Properties}
We prove that the ESWM mechanism satisfies the aforementioned four desirable economic properties.

\begin{lem}
	ESWM mechanism is always individually rational for all participants, except for unpunctual workers, regardless of the task submission time, $t^{sub}_i$.
\end{lem}
\begin{proof}
	To prove the individual rationality of the ESWM mechanism, we show that each participant's utility is non-negative by the end of an auction when reporting its true maximum valuation (requester) or cost (worker). \\
	\textbf{Requesters:}
	To prove the individual rationality of requesters, we show that $u_j^r \geq 0, \forall r_j \in R$. For a $r_j \in R_s$, his/her utility is $u_j^r = v_j(t_*^{sub}) - q_j$ as defined in (\ref{eq:1}). When $r_j \in R$ submits its true valuation $v_j^{max}$, the platform calculates the temporary fee, $q_j = (\alpha_j^{\beta}v^{max}_{th}\left|\Gamma_{j}\right|) / (\alpha_{th}^{\beta}\left|\Gamma_{th}\right|)  \leq v_j^{max}$ for $r_j$ in the winner selection process. Putting $q_j$ into (\ref{eq:1}), the utility of each winning requester is $u_j^r=v_j^{max}-q_j \geq 0$ before its matched worker submits the task result.  
	
	In the pricing step, the platform finally determines the fee for each winning requester as $q'_j$, based on the achieved valuation $v_j(t_i^{sub})$  by the matched worker $w_i$. Given $v_j(t_i^{sub})$  and $q'_j$, the utility of each winning requester becomes 
	\begin{equation}
	\label{eq:17}
	u_j^r =  v_j(t_i^{sub}) - q'_j =  v_j(t_i^{sub})(1 - \frac{q_j}{v_j^{max}}).
	\end{equation}
	Since $v_j(t_i^{sub})\geq 0$ as defined in (\ref{eq:5}) and $q_j \leq v_j^{max}$, $u_j^r$ is always non-negative. Therefore, the non-negative utility of each winning requester is guaranteed regardless of the task submission time. For the unselected requesters, their utility is 0 as defined in (\ref{eq:1}). Such non-negative utility proves that our incentive mechanism achieves the individual rationality of all requesters, regardless of the task submission time.\\	
	\textbf{Workers:}
	When a worker submits its true cost value, the platform calculates the temporary payment to each winning worker, $p_i =(c_{th}\lambda_i^{\beta})/\lambda_{th}^{\beta} \geq c_i$. Putting $p_i$ into (\ref{eq:2}), the utility of each winning worker before its task result submission is $u_i^p = p_i - c_i \geq 0$. 
	
	In the pricing step, the platform determines the final payment $p'_i$ to each winning worker $w_i \in W_s$ who is matched to $r_j \in R_s$, based on the achieved valuation $v_j(t_i^{sub})$. Given $v_j(t_i^{sub})$ and $p'_i$, the utility of $w_i \in W_s$ becomes
	\begin{equation}
	\label{eq:18}
	u_i^p = p'_i - c_i = \frac{v_j(t_i^{sub})}{v_j^{max}}p_i - c_i = \frac{v_j(t_i^{sub})}{v_j^{max}}(p_i - \frac{v_j^{max}}{v_j(t_i^{sub})}c_i). 
	\end{equation}
	
	Unlike the case of the winning requesters, our mechanism does not guarantee tardy workers non-negative utility. For punctual workers who meet the deadline of their requested task, i.e., $t_i^{sub} \leq t_j^d$, the non-negative utility is guaranteed since $v_j(t_i^{sub})= v_j^{max}$ if $t_i^{sub} \leq t_j^d$, which makes $u_i^p = p_i - c_i \geq 0$. 	
	In contrast, for unpunctual workers who submit their requested task result past the deadline, i.e., $t_i^{sub} > t_j^d$, the non-negative utility is not guaranteed. When $t_i^{sub} > t_j^d$, we cannot assure that $p_i - (v_j^{max}c_i)/v_j(t_i^{sub}) \geq 0$ in (\ref{eq:18}) since $v_j(t_i^{sub}) < v_j^{max}$. For the unselected workers, their utility is 0 as defined in (\ref{eq:2}). Such non-negative utility \textit{exclusively} for punctual workers can be an additional incentive to promote workers to meet the deadline.
\end{proof}

\begin{lem}
	ESWM mechanism is budget-balanced at least before the task submission.
\end{lem}
\begin{proof}
	In Algorithm. \ref{Matching1}, the ESWM mechanism checks whether $\sum_{p_i \in P}p_i > \sum_{q_j \in Q}$. If the condition is met, the platform sets $R_s$, $W_s$, $Q$, and $P$ empty sets and returns them to terminate the crowdsourcing process. Otherwise, it continues the crowdsourcing process. Therefore, the platform utility is always non-negative, which achieves the budget-balance of our mechanism before the task submission.
\end{proof}
\begin{lem}
	ESWM mechanism is computationally efficient.
\end{lem}
\begin{proof}
	In the ESWM mechanism, the winner selection algorithm, the matching algorithm, and the pricing algorithm run upper-bounded by $\mathcal{O}(\max \{N, M\} K)$, $\mathcal{O}(K)$, and $\mathcal{O}(K)$ respectively. Thus, the time complexity of our incentive mechanism is bounded in $\mathcal{O}(\max\{N, M\} K)$, which achieves the computational efficiency.
\end{proof}

\begin{lem}
	ESWM mechanism is truthful.
\end{lem}
\begin{proof}
	To prove the truthfulness of our incentive mechanism, we use Myerson's Theorem \cite{myerson}. According to \cite{myerson}, we need to show that our mechanism satisfies two conditions, \textit{monotonicity} of the winner selection process and \textit{critical value}-based pricing to the winners to prove its truthfulness.
	
	A monotonicity of mechanisms is satisfied if a seller $w_i$ (a buyer $r_j$) wins the auction by bidding $c_i$ ($v_j^{max}$), it will surely win the auction by bidding ${c'}_i \leq c_i$ (${v'}^{max}_j \geq v_j^{max}$). 
	The critical value is the maximum (minimum) value that a seller (buyer) can ask to win the auction. In other words, if a seller (buyer) bids higher (lower) value than the critical value, it will lose the auction. As our incentive mechanism is in a form of double auction, we prove truthfulness by showing the truthfulness of requesters and workers, respectively. \\
	\textbf{Requesters:} We prove the truthfulness of requesters by showing that the WRSA satisfies monotonicity and the temporary fee $q_j$ is the critical value. 
	In the WRSA, monotonicity is evident. If a requester $r_j$  wins the auction by bidding $v_j^{max}$, it means that we have a threshold requester $r_{th}$, such that $v_j^{max}/(\alpha_j^{\beta}\left|\Gamma_j\right|) \geq v^{max}_{th}(\alpha_{th}^{\beta}\left|\Gamma_{th} \right|)$. 
	Then, if the requester $r_j$  bids ${v'}^{max}_j$ such that ${v'}^{max}_j \geq v_j^{max}$, it will definitely win the auction as ${v'}^{max}_j/(\alpha_j^{\beta}\left|\Gamma_j\right|) \geq v_j^{max}/(\alpha_j^{\beta}\left|\Gamma_j\right|) \geq v^{max}_{th}/(\alpha_{th}^{\beta} \left|\Gamma_{th} \right|)$. Therefore, the WRSA satisfies monotonicity in the winner selection.
	
	To verify the critical value based pricing, we show that if a requester $r_j$ submits ${v"}^{max}_j$ which is less than the temporary fee $q_j$, then the requester will lose the auction. In the WRSA, if a requester $r_j$ bids ${v"}^{max}_j$ such that ${v"}^{max}_j < q_j$, then it will be ousted from $R_s$ as ${v"}^{max}_j/(\alpha_j^{\beta}\left|\Gamma_{j}\right|) < q_j/(\alpha_j^{\beta}\left|\Gamma_{j}\right|)= v^{max}_{th}/(\alpha_{th}^{\beta} \left|\Gamma_{th}\right|)$. Therefore, $q_j$ is the critical value for $r_j$. Since both monotonicity of the winner selection and critical value-based pricing to the winners are satisfied, the WRSA achieves truthfulness. 
	\\
	\textbf{Workers:} As in the case of requesters, we prove the truthfulness of workers by showing that the WWSA satisfies monotonicity and $p_i$ for $w_i \in W_s$ is the critical value. If a worker $w_i$ wins the auction, it means that we have a threshold worker, such that $c_i/\lambda_i^{\beta} \leq c_{th}/\lambda_{th}^{\beta}$. Then, if the worker $w_i$ asks for $c'_i$ such that $c_i' \leq c_i$, it will surely win the auction as $c_i'/\lambda_i^{\beta} \leq c_i/\lambda_i^{\beta} \leq c_{th}/\lambda_{th}^{\beta}$. This shows that the WWSA satisfies monotonicity in the winner selection.
	
	To verify the critical value based pricing, we show that if $w_i$ submits ${c"}_i$ which is larger than $p_i$, then the worker will lose the auction. In the WWSA, if a worker $w_i$ submits ${c"}_i$ such that ${c"}_i > p_i$, then it will be excluded from $W_s$ because ${c"}_i/\lambda_i^{\beta} > p_i/\lambda_i^{\beta}= c_{th}/\lambda_{th}^{\beta}$. Therefore, $p_i$ is the critical value for $w_i$. Since both monotonicity of the winner selection and critical value-based pricing to the winners are satisfied, the WWSA achieves truthfulness.
\end{proof}

\section{Conclusion}
\label{part:6}
In this work, we addressed the impractical assumption in the existing works that all the workers meet the deadline and task valuation remains constant. To bridge the gap of such imperfect punctuality and the consequential potential task depreciation in crowdsourcing systems, we modeled the workers' heterogeneous punctuality behavior and task depreciation. To address the limitations of the existing works, we proposed an expected social welfare maximizing (ESWM) mechanism that selects appropriate requester-worker pairs considering heterogeneity in task depreciation speed and workers' punctuality. Different from the existing works, the ESWM mechanism aims to achieve higher social welfare in a long-term for multiple rounds of auction by attracting and retaining more participants, rather than attempting to simply maximize the social welfare for one round. In the evaluation, we compared the ESWM mechanism to one of the existing works in both short-term and long-term scenarios to reflect the continuous competition between crowdsourcing service platforms. Simulation results show that the ESWM mechanism achieves higher social welfare and platform utility than the benchmark. Moreover, we proved that the ESWM mechanism achieves the desirable economic properties.

\section*{Acknowledgment}
This research was funded by the MSIT, Korea, under the "ICT Consilience Creative Program" (IITP-2017-R0346-16-1007) supervised by the IITP, by the KEIT, Korea, under the "Global Advanced Technology Center" (10053204), and by the NRF, Korea, under the "Basic Science Research Program" (NRF-2015R1C1A1A01053788). 
 
\bibliographystyle{IEEEtran}
\bibliography{RPE}
 
\begin{IEEEbiography}[{\includegraphics[width=1in,height=1.25in,clip,keepaspectratio]{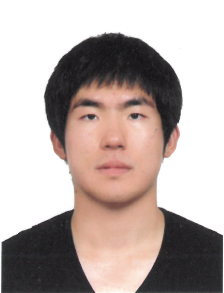}}]{Duin Baek} 
    received his B.Sc. degree from Yonsei University, Seoul in electrical and electronics engineering. He is currently a Ph.D. candidate in the Department of Computer Science at Stony Brook University and the State University of New York Korea. His current research includes designing mechanisms for crowdsourcing service, distributed computing via IoT systems, and computer vision for VR/AR.
\end{IEEEbiography}

\begin{IEEEbiography}[{\includegraphics[width=1in,height=1.25in,clip,keepaspectratio]{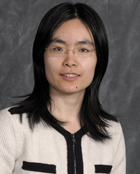}}]{Jing Chen} 
    is an Assistant Professor in the Department of Computer Science at Stony Brook University. She is an Affiliated Assistant Professor in the Department of Economics and an affiliated member of the Center for Game Theory. Before joining Stony Brook, she was a postdoctoral fellow in the School of Mathematics at the Institute for Advanced Study, Princeton. She received her Ph.D. in Computer Science from Massachusetts Institute of Technology in 2012, M.E. and B.E. in Computer Science from Tsinghua University, China. Her main research interests are computational game theory, mechanism design, auctions, markets, and healthcare. She is also interested in algorithms and computational complexity.
\end{IEEEbiography}

\begin{IEEEbiography}[{\includegraphics[width=1in,height=1.25in,clip,keepaspectratio]{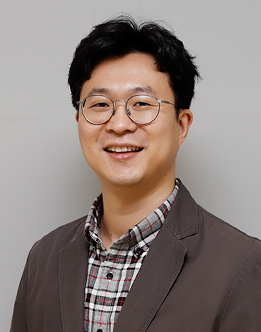}}]{Bong Jun Choi} 
    received his B.Sc. and M.Sc. degrees from Yonsei University, Korea, both in electrical and electronics engineering, and the Ph.D. degree from University of Waterloo, Canada, in electrical and computer engineering. He is an assistant professor at the Department of Computer Science, State University of New York Korea, Korea, and concurrently a research assistant professor at the Department of Computer Science, Stony Brook University, USA. His research focuses are energy efficient networks, distributed mobile wireless networks, smart grid communications, and network security.
\end{IEEEbiography}




\end{document}